\documentclass[acmsmall]{acmart}

\usepackage{booktabs}
\usepackage{multirow}
\usepackage{subfig}

\AtBeginDocument{%
  }

\setcopyright{acmlicensed}
\copyrightyear{2026}
\acmYear{2026}
\acmDOI{XXXXXXX.XXXXXXX}

\acmJournal{TOIS}
\acmVolume{37}
\acmNumber{4}
\acmArticle{111}
\acmMonth{9}

\begin{document}

\title{Mitigating Popularity Bias in Recommendation with Global Listwise Learning and Progressive Bi-Weighting}


\author{Tianyu Zhu}
\authornote{This work was partly conducted at the University of Montreal.}
\email{ztybuaa@126.com}
\affiliation{%
  \institution{School of Economics and Management, Beihang University}
  \city{Beijing}
  \country{China}
}

\author{Jiandong Ding}
\email{jdding@fudan.edu.cn}
\affiliation{%
  \institution{College of Computer Science and Artificial Intelligence, Fudan University}
  \city{Shanghai}
  \country{China}
}

\author{Yansong Shi}
\authornote{Corresponding author.}
\email{shiys@fudan.edu.cn}
\affiliation{%
  \institution{School of Management, Fudan University}
  \city{Shanghai}
  \country{China}
}

\author{Guoqing Chen}
\email{chengq@sem.tsinghua.edu.cn}
\affiliation{%
  \institution{School of Economics and Management, Tsinghua University}
  \city{Beijing}
  \country{China}
}

\author{Jian-Yun Nie}
\email{nie@iro.umontreal.ca}
\affiliation{%
  \institution{Department of Computer Science and Operations Research, University of Montreal}
  \city{Montreal}
  \country{Canada}
}

\renewcommand{\shortauthors}{Zhu et al.}

\begin{abstract}
    In recommender systems, user feedback typically follows a long-tail distribution, which leads many recommendation algorithms to exacerbate popularity bias by disproportionately favoring popular items. To mitigate this issue, recent studies have employed Inverse Propensity Scoring (IPS) to rebalance training data via reweighting user-item interactions. However, the effectiveness of IPS-based approaches is often constrained by locally unbiased objectives and inaccurate propensity estimation. In this paper, we propose Multinomial Likelihood with Bi-Weighting (Mult-BiW) to address these limitations. First, we introduce a debiasing framework, termed Mult-IPS, which integrates multinomial likelihood with IPS to capture global and unbiased user preferences over the entire item set. Second, we develop a Bi-Weighting (BiW) strategy that jointly leverages propensity scores and a collection model, incorporating a smoothing mechanism to enhance the robustness of propensity estimation. We further provide theoretical analyses that establish an upper bound on the empirical bias and characterize the optimal form of the collection model. Third, to mitigate the adverse effects of aggressive reweighting on representation learning, we design a Progressive Bi-Weighting strategy that gradually transitions from discriminative representation learning to popularity debiasing. Extensive experiments on real-world datasets show that Mult-BiW consistently outperforms state-of-the-art baselines.
\end{abstract}



\received{13 May 2024}
\received[revised]{26 November 2024}
\received[revised]{29 January 2026}
\received[accepted]{22 September 2026}

\maketitle

\section{Introduction}
Recommender systems have experienced rapid growth in recent years and have become a core component of many online platforms, including e-commerce, location-based services, and social media. By modeling users’ historical behaviors, these systems aim to provide personalized content that aligns with users’ latent interests. However, a fundamental challenge in real-world recommender systems arises from the long-tail distribution of user feedback, where a small fraction of popular items accounts for the majority of interactions, while the vast majority of tail items receive limited exposure \cite{canamares2018should,abdollahpouri2020multi}. This severe imbalance in item popularity often induces popularity bias in learned models. When trained on such skewed data, recommendation algorithms tend to propagate and even amplify the bias by over-recommending popular items \cite{abdollahpouri2017controlling,abdollahpouri2020connection}, thereby undermining personalization quality, serendipity, and fairness \cite{chen2023bias}.

To mitigate popularity bias, a variety of debiasing techniques have been proposed. One line of work focuses on \textit{ranking adjustment}, which heuristically calibrates item scores or rankings by compensating for item popularity \cite{abdollahpouri2017controlling,zhu2021popularity}. Despite their empirical effectiveness, these approaches typically lack a principled theoretical foundation and may yield suboptimal or unstable results \cite{zhang2021causal}. More recently, causal inference has inspired methods like DICE \cite{zheng2021disentangling}, which utilizes \textit{causal embeddings} to disentangle user interest from conformity effects. However, DICE exhibits limited debiasing capability, as it only models relative popularity differences between item pairs and does not explicitly address global popularity imbalance. 

Inverse Propensity Scoring (IPS) \cite{schnabel2016recommendations,ma2019missing} has emerged as a more principled and effective debiasing paradigm. By reweighting user-item interactions according to their propensities, i.e., exposure probabilities, IPS aims to correct the mismatch between observed and ideal data distributions during training. Although IPS-based methods have demonstrated strong debiasing performance, they suffer from two critical issues that remain largely unaddressed in existing studies.

\textit{Limitations of Locally Unbiased Objectives}.  
Existing IPS-based methods typically adopt pointwise or pairwise loss functions. For instance, Rel-MF \cite{saito2020unbiased} proposes an unbiased pointwise loss based on logistic or Gaussian likelihoods, while UBPR \cite{saito2020unbiased2} derives an unbiased pairwise ranking objective from Bayesian Personalized Ranking (BPR) \cite{rendle2009bpr}. These approaches focus on achieving local unbiasedness by correcting individual interactions or pairwise item comparisons. However, they are inherently limited in capturing users’ global preference structures over the entire item set \cite{liang2018variational,chen2020efficient2}. This limitation is further exacerbated when IPS is incorporated, as the long-tail distribution of item popularity can induce highly skewed weights, causing the model to overemphasize certain items. This observation motivates the need for a more suitable loss function that enables IPS-based methods to learn globally unbiased preferences in item recommendation.

\textit{Sensitivity to Propensity Misspecification}.  
IPS relies on accurate propensity values to effectively correct popularity bias. In practice, however, true propensities are unobservable and must be estimated. Most existing methods approximate propensities using item popularity statistics \cite{saito2020unbiased,saito2020unbiased2}. Such estimators tend to systematically underestimate the exposure probabilities of tail items \cite{saito2020asymmetric}, since true propensities may be influenced by factors beyond popularity, such as social effects or user exploration behaviors. Inaccurate propensity estimation can introduce substantial empirical bias into IPS estimators, thereby severely degrading recommendation performance \cite{chen2023bias}. This challenge highlights the necessity of more principled and robust propensity estimation methods with theoretical guarantees.

In this paper, we propose a novel approach, termed \textit{Multinomial Likelihood with Bi-Weighting} (Mult-BiW), to address the aforementioned challenges. Mult-BiW consists of three key components. First, we introduce a listwise unbiased learning framework, named Mult-IPS, which integrates IPS with the multinomial likelihood to model global and unbiased user preferences over the item set. By leveraging the softmax function under a constrained probability budget, the multinomial likelihood encourages competition among items and naturally captures global ranking information for each user. Moreover, Mult-IPS is model-agnostic and can be seamlessly applied to a wide range of recommendation backbones. Second, we propose a Bi-Weighting (BiW) strategy for improving propensity estimation. Inspired by the \textit{Jelinek--Mercer} smoothing technique in language modeling \cite{zhai2004study}, BiW combines estimated propensities with a collection model to enhance robustness, effectively increasing the weights of tail items while attenuating those of popular items. We further provide theoretical analyses that establish an upper bound on the empirical bias induced by inaccurate propensities and derive the optimal form of the collection model in BiW. Third, recognizing that aggressive rebalancing from the early training stages may impair representation learning \cite{zhou2020bbn}, we design a Progressive Bi-Weighting (PBiW) strategy that gradually increases the influence of propensity weights over training epochs. This \textit{cumulative learning} scheme allows the model to smoothly shift from discriminative representation learning to effective popularity debiasing.

Our main contributions are summarized as follows:
\begin{itemize}
    \item We propose a listwise unbiased learning framework that employs IPS-weighted multinomial likelihood to capture global user preferences over the entire item set, addressing the limitations of local preference learning in existing pointwise and pairwise debiasing methods. The framework is model-agnostic and broadly applicable.
    \item We introduce a Bi-Weighting strategy that improves propensity estimation by smoothing propensities with a collection model. We provide theoretical analyses to characterize the optimal collection model and further propose a Progressive Bi-Weighting scheme to mitigate the negative impact of data rebalancing on representation learning.
    \item We conduct extensive experiments on six real-world datasets, demonstrating that Mult-BiW consistently outperforms state-of-the-art debiasing baselines and yields more accurate and balanced recommendations.
\end{itemize}

The remainder of this paper is organized as follows. Section~\ref{sec:pre} introduces the preliminaries, Section~\ref{sec:model} details the proposed method, Section~\ref{sec:exp} presents experimental evaluations, and Section~\ref{sec:rel} reviews related work. Finally, Section~\ref{sec:con} concludes the paper and outlines future research directions.

\section{Preliminaries}
\label{sec:pre}

\subsection{Problem Formulation}
We first introduce the basic concepts and definitions relevant to our study. The notation used throughout this paper is summarized in Table~\ref{tab:notation} and will be specified as needed. We begin with the conventional top-$N$ recommendation problem \cite{deshpande2004item}.

\begin{definition}[Top-$N$ Recommendation]
Let $\mathcal{U}=\{u_1,u_2,\ldots,u_{|\mathcal{U}|}\}$ denote the set of users, $\mathcal{I}=\{i_1,i_2,\ldots,i_{|\mathcal{I}|}\}$ denote the set of items, and $\mathbf{Y}\in\mathbb{R}^{|\mathcal{U}|\times|\mathcal{I}|}$ denote the binary user-item interaction matrix, where $y_{ui}=1$ indicates that user $u$ has interacted with item $i$. The top-$N$ recommendation task aims to predict preference scores for unobserved user-item pairs and recommend the top-$N$ items with the highest predicted scores for each user.
\end{definition}

In practice, recommendation models are typically trained and evaluated on randomly split interaction data, which inherently reflects the long-tail distribution of item popularity. As a result, both training and testing sets are biased toward popular items. To properly assess the effectiveness of popularity debiasing methods, evaluation should be conducted on unbiased or balanced test data. Following common practice \cite{liang2016causal,bonner2018causal,wei2021model}, we adopt a data splitting strategy that constructs balanced test sets with a uniform item distribution. Accordingly, we define the popularity-debiased recommendation problem as follows \cite{zheng2021disentangling}.

\begin{definition}[Popularity-Debiased Recommendation]
Popularity-debiased recommendation aims to train recommendation models on biased training data $\mathcal{D}_{\mathrm{train}}$, in which item interactions follow a long-tail distribution, and to evaluate them on unbiased test data $\mathcal{D}_{\mathrm{test}}$, where items are uniformly distributed. The objective is to maximize recommendation performance on $\mathcal{D}_{\mathrm{test}}$ under non-IID popularity distributions in the training and test data.
\end{definition}

In this work, we focus on addressing the popularity-debiased top-$N$ recommendation problem by developing principled learning and reweighting strategies.

\begin{table}[!t]
    \centering
    \caption{Notation.}
    \begin{tabular}{ll}
        \toprule
        Symbol & Description \\
        \midrule
        $\mathcal{U}$ & Set of users \\
        $\mathcal{I}$ & Set of items \\
        $\mathbf{Y}$ & User-item interaction matrix \\
        $\mathbf{R}$ & User-item relevance matrix \\
        $\hat{\mathbf{R}}$ & Predicted user-item relevance matrix \\
        $\gamma_{ui}$ & Probability that item $i$ is relevant to user $u$, i.e., $P(r_{ui}=1$) \\
        $\theta_{ui}$ & Probability that item $i$ is exposed to user $u$, i.e., $P(y_{ui}=1|r_{ui}=1)$ \\
        $\hat{\theta}_{ui}$ & Estimated propensity of user-item pair $(u,i)$ \\
        $\delta_{ui}$ & Local loss for user-item pair $(u,i)$ \\
        $p_i$ & Popularity of item $i$ \\
        $\alpha$ & Propensity balancing hyperparameter \\
        $T$ & Current training epoch \\
        $T_{max}$ & Total number of training epochs \\
        $\eta$ & Cumulative learning hyperparameter \\
        \bottomrule
    \end{tabular}
    \label{tab:notation}
\end{table}

\subsection{IPS for Popularity Debiasing in Item Recommendation}
We briefly review the Inverse Propensity Scoring (IPS) framework for popularity debiasing in item recommendation \cite{saito2020unbiased}. Common recommendation metrics, such as Recall and NDCG, are defined with respect to the interaction probability $P(y_{ui}=1)$. However, observed interactions do not always faithfully reflect users’ true interests, as they are confounded by exposure and popularity effects. Consequently, recommendation policies should ideally be evaluated based on item relevance to users rather than observed interactions alone \cite{yang2018unbiased}.

Under this perspective, an ideal pointwise loss function that directly optimizes user relevance can be formulated as \cite{saito2020unbiased}:
\begin{equation}
\label{equ:ideal}
\mathcal{L}_{\mathrm{ideal}}(\hat{\mathbf{R}})
= \frac{1}{|\mathcal{U}|}
\sum_{u\in\mathcal{U}}\sum_{i\in\mathcal{I}}
\gamma_{ui}\,\delta_{ui}^{(1)} + (1-\gamma_{ui})\,\delta_{ui}^{(0)},
\end{equation}
where $\gamma_{ui}=P(r_{ui}=1)$ denotes the probability that item $i$ is relevant to user $u$, $r_{ui}\in\{0,1\}$ represents the unobserved relevance label, and $\delta_{ui}^{(r_{ui})}$ denotes the local loss for the user-item pair $(u,i)$ given relevance $r_{ui}$. A typical choice for $\delta_{ui}^{(r_{ui})}$ is the binary cross-entropy loss \cite{saito2020unbiased}, i.e., $\delta_{ui}^{(r_{ui})}=-r_{ui}\log\sigma(\hat{r}_{ui})-(1-r_{ui})\log(1-\sigma(\hat{r}_{ui}))$.

Since relevance labels are unobservable, IPS introduces propensities to construct an unbiased estimator of the ideal loss.

\begin{definition}[Propensity]
The propensity of a user-item pair $(u,i)$, denoted by $\theta_{ui}$, is defined as the conditional probability $\theta_{ui}=P(y_{ui}=1\!\mid\! r_{ui}=1)$, representing the likelihood that a relevant item is observed (i.e., interacted with) by the user.
\end{definition}

Using propensities, the IPS-based unbiased estimator of the ideal loss can be expressed as:
\begin{equation}
\label{equ:unbiased}
\mathcal{L}_{\mathrm{unbiased}}(\hat{\mathbf{R}}\!\mid\!\mathbf{\Theta})
= \frac{1}{|\mathcal{U}|}
\sum_{u\in\mathcal{U}}\sum_{i\in\mathcal{I}}
\frac{y_{ui}}{\theta_{ui}}\,\delta_{ui}^{(1)}
+ \left(1-\frac{y_{ui}}{\theta_{ui}}\right)\delta_{ui}^{(0)}.
\end{equation}

Intuitively, IPS compensates for popularity-induced exposure bias by reweighting observed interactions during training. Nevertheless, this estimator is inherently pointwise and treats each item independently, which may limit its ability to optimize ranking quality and capture users’ global preference structures. We address this limitation in the following sections.

\section{Methodology}
\label{sec:model}

\begin{figure}[!t]
    \centering
    \includegraphics[width=\linewidth]{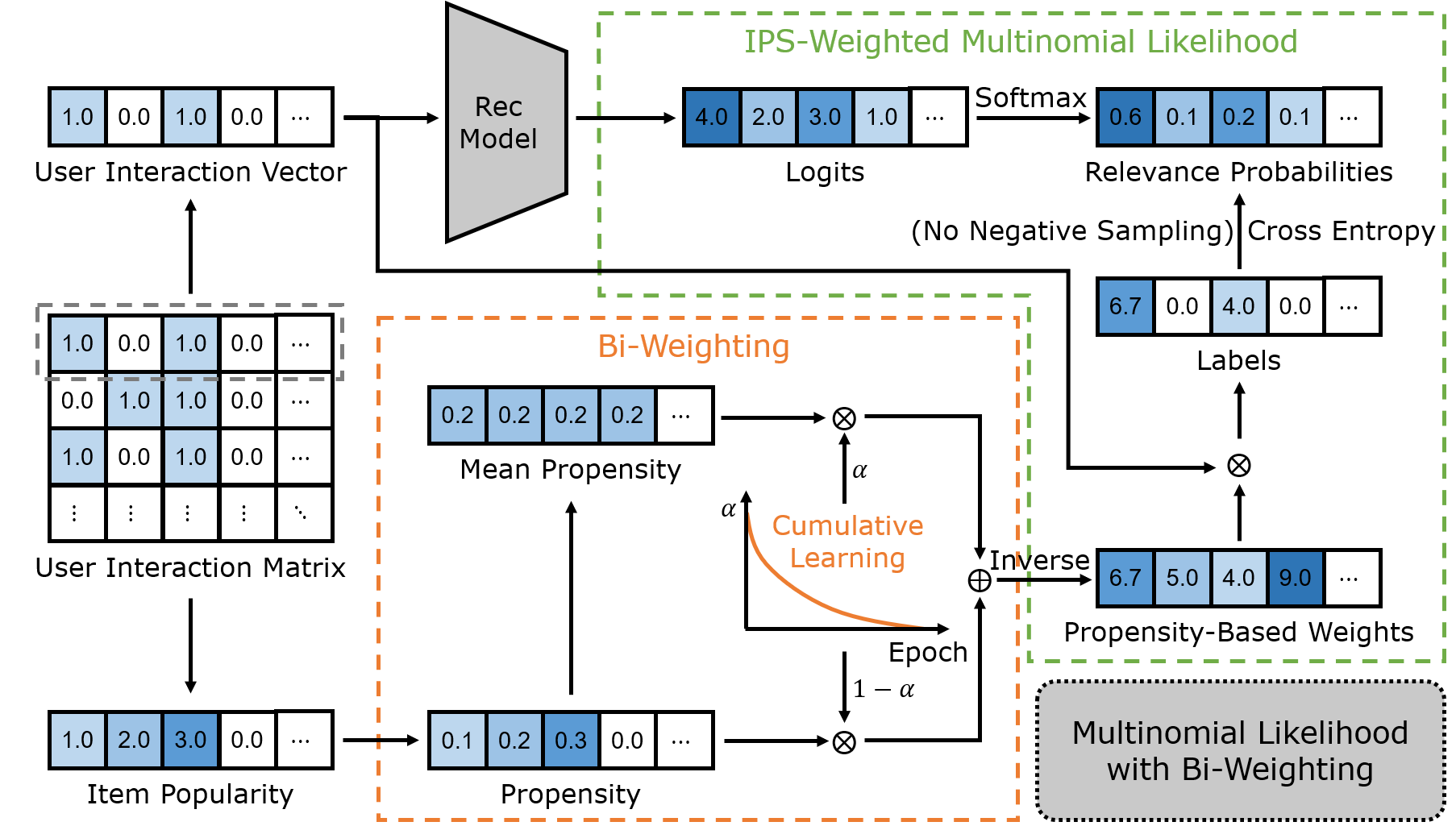}
    \caption{Overview of the proposed Mult-BiW framework. It comprises three key components: (a) IPS-weighted multinomial likelihood for listwise unbiased preference estimation, (b) bi-weighting for improved propensity modeling, and (c) cumulative learning for progressive debiasing.}
    \label{fig:framework}
\end{figure}

To address the limitations of existing popularity debiasing methods, we propose a novel approach termed \textit{Multinomial Likelihood with Bi-Weighting} (Mult-BiW). As illustrated in Figure~\ref{fig:framework}, Mult-BiW consists of three key components: (a) an IPS-weighted multinomial likelihood for listwise unbiased preference learning, (b) a Bi-Weighting strategy for robust propensity estimation, and (c) a cumulative learning scheme for progressive popularity debiasing.

\subsection{IPS-Weighted Multinomial Likelihood}
For each user $u$, we assume that the relevance vector over the item set, denoted by $\mathbf{r}_u$, is generated from a multinomial distribution \cite{liang2018variational}:
\begin{equation}
\mathbf{r}_u \sim \mathrm{Mult}(N_u, \tilde{\mathbf{r}}_u),
\end{equation}
where $N_u=\sum_i r_{ui}$ represents the number of relevant items for user $u$, and $\tilde{\mathbf{r}}_u=\mathrm{softmax}(\hat{\mathbf{r}}_u)$ denotes the predicted relevance distribution obtained by normalizing the model outputs $\hat{\mathbf{r}}_u$ via the softmax function.

The proposed framework is model-agnostic: the relevance scores $\hat{\mathbf{r}}_u$ can be produced by a wide range of recommendation models, including matrix factorization \cite{koren2009matrix}, autoencoder-based methods \cite{liang2018variational}, and graph convolutional networks \cite{he2020lightgcn}. Conditioned on the predicted relevance distribution, the log-likelihood for user $u$ is given by:
\begin{equation}
\log P(\mathbf{r}_u \!\mid\! \hat{\mathbf{r}}_u)
= \sum_{i\in\mathcal{I}} r_{ui} \log \tilde{r}_{ui}.
\end{equation}

Since the marginal relevance probability is defined as $\gamma_{ui}=P(r_{ui}=1)$, the ideal loss function based on the negative multinomial log-likelihood can be expressed as:
\begin{equation}
\mathcal{L}_{\mathrm{Mult}}(\hat{\mathbf{R}})
= -\frac{1}{|\mathcal{U}|}
\sum_{u\in\mathcal{U}}\sum_{i\in\mathcal{I}}
\gamma_{ui} \log \tilde{r}_{ui}.
\end{equation}

To incorporate the multinomial likelihood into the IPS framework introduced in Equation~(\ref{equ:unbiased}), we define the local loss as the negative multinomial log-likelihood, i.e.,
$\delta_{ui}^{(r_{ui})} = - r_{ui}\log \tilde{r}_{ui}$.
Accordingly, $\delta_{ui}^{(1)} = -\log \tilde{r}_{ui}$ and $\delta_{ui}^{(0)} = 0$, which yields the proposed listwise unbiased estimator, referred to as the Mult-IPS estimator:
\begin{equation}
\label{equ:Mult-IPS}
\mathcal{L}_{\mathrm{Mult\mbox{-}IPS}}(\hat{\mathbf{R}}\!\mid\!\mathbf{\Theta})
= -\frac{1}{|\mathcal{U}|}
\sum_{u\in\mathcal{U}}\sum_{i\in\mathcal{I}}
\frac{y_{ui}}{\theta_{ui}} \log \tilde{r}_{ui}.
\end{equation}

\textit{Remark.}
We next discuss the advantages of Mult-IPS in comparison with existing IPS-based debiasing methods. Alternative likelihood functions have been widely adopted in recommender systems, such as the Gaussian likelihood \cite{mnih2008probabilistic,koren2009matrix}:
\begin{equation}
\log P(\mathbf{r}_u \!\mid\! \hat{\mathbf{r}}_u)
= -\frac{1}{2} \sum_{i\in\mathcal{I}} (r_{ui}-\hat{r}_{ui})^2,
\end{equation}
and the logistic likelihood \cite{johnson2014logistic,he2017neural}:
\begin{equation}
\log P(\mathbf{r}_u \!\mid\! \hat{\mathbf{r}}_u)
= \sum_{i\in\mathcal{I}}
r_{ui}\log\sigma(\hat{r}_{ui}) + (1-r_{ui})\log\bigl(1-\sigma(\hat{r}_{ui})\bigr),
\end{equation}
where $\sigma(\cdot)$ denotes the sigmoid function.

The multinomial likelihood offers two key advantages in the IPS setting. First, due to the probability simplex constraint imposed by the softmax function, the multinomial likelihood rewards the model for allocating higher probability mass to relevant items under a fixed budget \cite{liang2018variational}. As a result, it naturally acts as a listwise loss that induces competition among all items and captures global ranking information for each user. In contrast, Gaussian and logistic likelihoods are inherently pointwise, focusing on individual item predictions. Although unbiased pairwise objectives such as UBPR \cite{saito2020unbiased2} have been proposed, they only model relative preferences between item pairs and are insufficient for learning users’ global preference distributions \cite{chen2020efficient2}.

Second, the IPS estimator in Equation~(\ref{equ:unbiased}) assigns both positive and negative weights to observed interactions. When Gaussian or logistic likelihoods are used, the negative loss term $\delta_{ui}^{(0)}$ is non-zero. Combined with IPS reweighting, this may lead to large negative values caused by $1 - \frac{y_{ui}}{\theta_{ui}} < 0$ for $y_{ui}=1$, resulting in high variance and numerical instability during training \cite{saito2020unbiased2}. In contrast, the multinomial likelihood yields $\delta_{ui}^{(0)}=0$, thereby avoiding large negative contributions and leading to a more stable optimization process.

While softmax-based losses have been explored in recent debiasing studies \cite{zhang2023invariant,zhang2024robust}, our work explicitly investigates their suitability for IPS-based learning. Recent findings further suggest that the softmax loss explicitly accounts for uncertainty in the underlying distributions, enabling greater robustness to distributional shifts \cite{wu2024bsl}. This property is particularly beneficial in IPS-based debiasing, where mismatches between training and testing distributions are inevitable. Moreover, the softmax formulation implicitly regularizes model predictions by penalizing excessive variance, thereby promoting more balanced and equitable recommendation outcomes \cite{wu2024effectiveness}.

\subsection{Bi-Weighting}

We begin by analyzing the IPS estimator with the multinomial likelihood to establish a theoretical foundation for the proposed approach. Proposition~\ref{proposition1} shows that the Mult-IPS estimator is unbiased with respect to the ideal loss defined over the complete user-item relevance matrix.

\begin{proposition}[Unbiasedness of Mult-IPS Estimator]
\label{proposition1}
The Mult-IPS estimator is unbiased with respect to the multinomial loss:
\begin{equation}
\mathbb{E}\!\left[\mathcal{L}_{\mathrm{Mult\mbox{-}IPS}}(\hat{\mathbf{R}}\!\mid\!\mathbf{\Theta})\right]
= \mathcal{L}_{\mathrm{Mult}}(\hat{\mathbf{R}}).
\end{equation}
\end{proposition}

\begin{proof}
\begin{equation}
\label{eqn:proposition1}
\begin{aligned}
\mathbb{E}\!\left[\mathcal{L}_{\mathrm{Mult\mbox{-}IPS}}(\hat{\mathbf{R}}\!\mid\!\mathbf{\Theta})\right]
&= \mathbb{E}\!\left[-\frac{1}{|\mathcal{U}|}\sum_{u\in\mathcal{U}}\sum_{i\in\mathcal{I}}
\frac{y_{ui}}{\theta_{ui}}\log\tilde{r}_{ui}\right] \\
&= -\frac{1}{|\mathcal{U}|}\sum_{u\in\mathcal{U}}\sum_{i\in\mathcal{I}}
\frac{\mathbb{E}[y_{ui}]}{\theta_{ui}}\log\tilde{r}_{ui} \\
&= -\frac{1}{|\mathcal{U}|}\sum_{u\in\mathcal{U}}\sum_{i\in\mathcal{I}}
\gamma_{ui}\log\tilde{r}_{ui} \\
&= \mathcal{L}_{\mathrm{Mult}}(\hat{\mathbf{R}}).
\end{aligned}
\end{equation}
\end{proof}

In practice, true propensities $\mathbf{\Theta}$ are typically unknown and must be estimated, resulting in the use of an empirical loss, denoted by
$\hat{\mathcal{L}}_{\mathrm{Mult\mbox{-}IPS}}(\hat{\mathbf{R}}\!\mid\!\hat{\mathbf{\Theta}})$.
Lemma~\ref{lemma1} characterizes the bias introduced by inaccurate propensity estimation.

\begin{lemma}[Bias with Inaccurate Propensities]
\label{lemma1}
Let $\mathbf{\Theta}$ denote the true propensities and $\hat{\mathbf{\Theta}}$ their estimates, with $\hat{\theta}_{ui}>0$ for all $(u,i)$. The bias of the Mult-IPS estimator is given by:
\begin{equation}
\begin{aligned}
\Delta\!\left[\hat{\mathcal{L}}_{\mathrm{Mult\mbox{-}IPS}}(\hat{\mathbf{R}}\!\mid\!\hat{\mathbf{\Theta}})\right]
&= \left| \mathcal{L}_{\mathrm{Mult}}(\hat{\mathbf{R}})
- \mathbb{E}\!\left[\hat{\mathcal{L}}_{\mathrm{Mult\mbox{-}IPS}}(\hat{\mathbf{R}}\!\mid\!\hat{\mathbf{\Theta}})\right] \right| \\
&= \frac{1}{|\mathcal{U}|}
\left| \sum_{u\in\mathcal{U}}\sum_{i\in\mathcal{I}}
\left(\gamma_{ui}-\frac{\mathbb{E}\left[y_{ui}\right]}{\hat{\theta}_{ui}}\right)
\log\tilde{r}_{ui} \right|\\
&= \frac{1}{|\mathcal{U}|}
\left| \sum_{u\in\mathcal{U}}\sum_{i\in\mathcal{I}}
\left(1-\frac{\theta_{ui}}{\hat{\theta}_{ui}}\right)
\gamma_{ui}\log\tilde{r}_{ui} \right|.
\end{aligned}
\end{equation}
\end{lemma}

Lemma~\ref{lemma1} indicates that the bias of the empirical Mult-IPS estimator is governed by the discrepancy between the true and estimated propensities. This observation motivates improving propensity estimation to reduce empirical bias. 

In the absence of side information, existing methods typically estimate propensities using relative item popularity, $\hat{\theta}_{ui} = (p_i/\max_j p_j)^\beta$, where $p_i=\sum_{u}y_{ui}$ denotes the popularity of item $i$ and $\beta$ is usually set to $1$ \cite{yang2018unbiased,saito2020unbiased}. However, such estimators systematically underestimate the propensities of tail items \cite{saito2020asymmetric}, since true exposure probabilities may be influenced by factors beyond popularity, such as social effects or user exploration.

To improve robustness, we propose a \textit{Bi-Weighting} (BiW) strategy that smooths individual propensities using a collection-level prior. Inspired by \textit{Jelinek--Mercer} smoothing in language modeling \cite{zhai2004study}, we linearly interpolate between the estimated propensities and a collection model:
\begin{equation}
\hat{\theta}_{ui}^{\mathrm{BiW}}
= (1-\alpha)\hat{\theta}_{ui} + \alpha C,
\end{equation}
where $C$ denotes a collection-level propensity estimate and $\alpha\in[0,1]$ controls the degree of smoothing. The role of $C$ is to provide a global prior that attenuates overly small propensities for tail items while moderating large propensities for popular items.

Proposition~\ref{proposition2} shows that incorporating the collection model can reduce the empirical bias of Mult-IPS.

\begin{proposition}[Bias Reduction with BiW]
\label{proposition2}
Let $\hat{\mathbf{\Theta}}^{\mathrm{BiW}^*}$ denote the BiW propensities obtained by optimizing $\alpha$ on an unbiased validation set. Then,
\begin{equation}
\Delta\!\left[\hat{\mathcal{L}}_{\mathrm{Mult\mbox{-}IPS}}(\hat{\mathbf{R}}\!\mid\!\hat{\mathbf{\Theta}}^{\mathrm{BiW}^*})\right]
\le
\Delta\!\left[\hat{\mathcal{L}}_{\mathrm{Mult\mbox{-}IPS}}(\hat{\mathbf{R}}\!\mid\!\hat{\mathbf{\Theta}})\right].
\end{equation}
\end{proposition}

\begin{proof}
    \begin{equation}
    \label{eqn:proposition2}
        \begin{aligned}
            \Delta\!\left[\hat{\mathcal{L}}_{\mathrm{Mult\mbox{-}IPS}}(\hat{\mathbf{R}}\!\mid\!\hat{\mathbf{\Theta}}^{\mathrm{BiW}^*})\right]
            &= \min_{\alpha} \Delta\!\left[\hat{\mathcal{L}}_{\mathrm{Mult\mbox{-}IPS}}(\hat{\mathbf{R}}\!\mid\!\hat{\mathbf{\Theta}}^{\mathrm{BiW}})\right]\\
            &= \min_{\alpha} \frac{1}{|\mathcal{U}|} \left| \sum_{u \in \mathcal{U}}\sum_{i \in \mathcal{I}} \left( 1 - \frac{\theta_{ui}}{\hat{\theta}_{ui}^{\mathrm{BiW}}} \right) \gamma_{ui}\log \tilde{r}_{ui} \right|\\
            &\leq \frac{1}{|\mathcal{U}|} \left| \sum_{u \in \mathcal{U}}\sum_{i \in \mathcal{I}} \left( 1 - \frac{\theta_{ui}}{\hat{\theta}_{ui}} \right) \gamma_{ui}\log \tilde{r}_{ui} \right|\\
            &= \Delta\!\left[\hat{\mathcal{L}}_{\mathrm{Mult\mbox{-}IPS}}(\hat{\mathbf{R}}\!\mid\!\hat{\mathbf{\Theta}})\right].
        \end{aligned}
    \end{equation}
    The result follows directly from the fact that $\alpha=0$ is a feasible solution and that $\alpha$ is optimized to minimize the empirical bias.
\end{proof}

While BiW offers bias reduction, determining an appropriate collection model $C$ is nontrivial. Directly minimizing the bias expression in Lemma~\ref{lemma1} is intractable and depends on $\alpha$. Instead, we derive an upper bound on the bias and optimize $C$ with respect to this bound.

\begin{lemma}[Bias Upper Bound]
Assume that $0 < -\gamma_{ui}\log\tilde{r}_{ui} \le \delta$ and $\hat{\theta}_{ui}^{\mathrm{BiW}} \ge m$. Then,
\begin{equation}
\Delta\!\left[\hat{\mathcal{L}}_{\mathrm{Mult\mbox{-}IPS}}(\hat{\mathbf{R}}\!\mid\!\hat{\mathbf{\Theta}}^{\mathrm{BiW}})\right]
\le
\frac{\delta\sqrt{|\mathcal{I}|}}{m\sqrt{|\mathcal{U}|}}
\sqrt{\sum_{u\in\mathcal{U}}\sum_{i\in\mathcal{I}}
(\hat{\theta}_{ui}^{\mathrm{BiW}}-\theta_{ui})^2}.
\end{equation}
\end{lemma}

\begin{proof}
    \begin{equation}
        \begin{aligned}
            \Delta\!\left[\hat{\mathcal{L}}_{\mathrm{Mult\mbox{-}IPS}}(\hat{\mathbf{R}}\!\mid\!\hat{\mathbf{\Theta}}^{\mathrm{BiW}})\right]
            &= \frac{1}{|\mathcal{U}|} \left| \sum_{u \in \mathcal{U}}\sum_{i \in \mathcal{I}} \left( 1 - \frac{\theta_{ui}}{\hat{\theta}_{ui}^{\mathrm{BiW}}} \right) \gamma_{ui}\log \tilde{r}_{ui} \right|\\
            &\leq \frac{1}{|\mathcal{U}|} \sum_{u \in \mathcal{U}}\sum_{i \in \mathcal{I}} \left| \left( 1 - \frac{\theta_{ui}}{\hat{\theta}_{ui}^{\mathrm{BiW}}} \right) \gamma_{ui}\log \tilde{r}_{ui} \right|\\
            &\leq \frac{\delta}{m |\mathcal{U}|} \sum_{u \in \mathcal{U}}\sum_{i \in \mathcal{I}} \left| \hat{\theta}_{ui}^{\mathrm{BiW}} - \theta_{ui} \right|\\
            &\leq \frac{\delta \sqrt{|\mathcal{I}|}}{m \sqrt{|\mathcal{U}|}} \sqrt{\sum_{u \in \mathcal{U}}\sum_{i \in \mathcal{I}} (\hat{\theta}_{ui}^{\mathrm{BiW}} - \theta_{ui})^2}.
        \end{aligned}
    \end{equation}
    The last inequality follows from the application of the \textit{Cauchy--Schwarz} inequality.
\end{proof}

Minimizing this upper bound leads to the optimal value of $C$, as presented in Theorem~\ref{theorem}.

\begin{theorem}[Optimal Collection Model]
\label{theorem}
The optimal collection model for BiW is given by:
\begin{equation}
C^* = \frac{\sum_{i\in\mathcal{I}}\hat{\theta}_{ui}}{|\mathcal{I}|}.
\end{equation}
\end{theorem}

\begin{proof}
    Since
    \begin{equation}        
        \arg\min_{C} \frac{\delta |\mathcal{I}|}{m |\mathcal{U}|} \sqrt{\sum_{u \in \mathcal{U}}\sum_{i \in \mathcal{I}} (\hat{\theta}_{ui}^{\mathrm{BiW}} - \theta_{ui})^2}=\arg\min_{C} \sum_{u \in \mathcal{U}}\sum_{i \in \mathcal{I}} (\hat{\theta}_{ui}^{\mathrm{BiW}} - \theta_{ui})^2,
    \end{equation}
    by solving the above convex optimization problem, we have:
    \begin{equation}
    \label{eqn:optimum}
        C^* = \frac{\sum_{i \in \mathcal{I}}\left[ \theta_{ui} - (1 - \alpha) \hat{\theta}_{ui}\right]}{|\mathcal{I}|\alpha}.
    \end{equation}
    Note that the true propensity $\theta_{ui}$ in Equation~(\ref{eqn:optimum}) is unknown. To obtain a computable estimate of $C^*$, we assume $\theta_{ui}=\hat{\theta}_{ui}+e_{ui}$, where the error terms $e_{ui}$ are \textit{i.i.d.} with $\mathbb{E}(e_{ui}) = 0$. As the size of the item set (i.e., $|\mathcal{I}|$) is usually large, by the \emph{law of large numbers}, we have:
    \begin{equation}
        \begin{aligned}
            C^* 
            &= \frac{\sum_{i \in \mathcal{I}} \left[\hat{\theta}_{ui}+e_{ui}-(1 - \alpha) \hat{\theta}_{ui}\right]}{|\mathcal{I}|\alpha} \\
            &= \frac{\sum_{i \in \mathcal{I}} \hat{\theta}_{ui}}{|\mathcal{I}|} + \frac{\sum_{i \in \mathcal{I}} e_{ui}}{|\mathcal{I}|\alpha} \\
            &\rightarrow_p \frac{\sum_{i \in \mathcal{I}} \hat{\theta}_{ui}}{|\mathcal{I}|} + \frac{\mathbb{E}[e_{ui}]}{\alpha} \\
            &= \frac{\sum_{i \in \mathcal{I}} \hat{\theta}_{ui}}{|\mathcal{I}|}.
        \end{aligned}
    \end{equation}
\end{proof}

Accordingly, the BiW propensity estimator becomes:
\begin{equation}
\hat{\theta}_{ui}^{\mathrm{BiW}}
= (1-\alpha)\hat{\theta}_{ui}
+ \alpha \cdot \frac{\sum_{j\in\mathcal{I}}\hat{\theta}_{uj}}{|\mathcal{I}|}.
\end{equation}

Substituting this estimator into the Mult-IPS objective yields the proposed Mult-BiW estimator, which reweights user-item pairs using a combination of individual propensities and the average propensity:
\begin{equation}
\hat{\mathcal{L}}_{\mathrm{Mult\mbox{-}BiW}}(\hat{\mathbf{R}}\!\mid\!\hat{\mathbf{\Theta}})
=
-\frac{1}{|\mathcal{U}|}\sum_{u\in\mathcal{U}}\sum_{i\in\mathcal{I}}
\frac{y_{ui}}
{(1-\alpha)\hat{\theta}_{ui}
+ \alpha|\mathcal{I}|^{-1}\sum_{j\in\mathcal{I}}\hat{\theta}_{uj}}
\log\tilde{r}_{ui}.
\end{equation}

The Mult-BiW estimator acts as a generalized mixture model. As $\alpha \to 1$, Mult-BiW converges to a naive estimator assigning uniform weight to all items. Conversely, as $\alpha \to 0$, it recovers the standard Mult-IPS estimator. Overall, BiW performs global smoothing of propensities, increasing the robustness of tail-item estimates while suppressing the dominance of popular items \cite{zhai2004study}.

\subsection{Progressive Bi-Weighting}
The choice of the smoothing parameter $\alpha$ plays a critical role in the effectiveness of Bi-Weighting. Recent studies on class-imbalanced learning have shown that applying rebalancing strategies from the early stages of training may hinder representation learning, ultimately degrading model performance \cite{kang2019decoupling,zhou2020bbn,he2021rethinking}. Motivated by these findings, we adopt a \textit{cumulative learning} perspective that jointly accounts for representation learning and popularity debiasing.

Specifically, we propose a \textit{Progressive Bi-Weighting} (PBiW) strategy that gradually increases the influence of propensity-based reweighting across training epochs. In the early stages, the model focuses primarily on learning discriminative representations from the biased data distribution, and progressively shifts toward popularity debiasing as training proceeds. Formally, we schedule $\alpha$ as:
\begin{equation}
\alpha = 1 - \left(\frac{T}{T_{\max}}\right)^{\eta},
\end{equation}
where $T$ denotes the current training epoch, $T_{\max}$ is the total number of epochs, and $\eta$ controls the pace of the cumulative learning process \cite{zhou2020bbn}.

Intuitively, PBiW enables the model to dynamically balance representation learning and debiasing, leading to improved accuracy on tail items without sacrificing overall recommendation quality \cite{he2021rethinking}. The parameter $\eta$ can be adjusted to accommodate different backbone models, which may vary in their representational capacity and convergence behavior. We provide an empirical analysis of the impact of $\eta$ in Section~\ref{sec:eta}.

\subsection{Discussion}

\subsubsection{Comparison with Propensity Clipping}

A widely used technique for improving IPS is \textit{propensity clipping}, which mitigates the high variance of IPS estimators by enforcing a lower bound on propensity values \cite{gilotte2018offline,gruson2019offline,saito2020unbiased}. When combined with the multinomial likelihood, the clipped IPS objective (Mult-CIPS) can be written as:
\begin{equation}
\hat{\mathcal{L}}_{\mathrm{Mult\mbox{-}CIPS}}(\hat{\mathbf{R}}\!\mid\!\hat{\mathbf{\Theta}})
= -\frac{1}{|\mathcal{U}|}
\sum_{u\in\mathcal{U}}\sum_{i\in\mathcal{I}}
\frac{y_{ui}}{\max\{\hat{\theta}_{ui},M\}}
\log\tilde{r}_{ui},
\end{equation}
where $M\in(0,1]$ is a predefined clipping threshold.

By truncating small propensity values, clipping effectively reduces estimator variance \cite{saito2020unbiased}. However, its impact is confined to a limited subset of tail items determined by the hard threshold $M$. In contrast, the proposed Bi-Weighting strategy performs a soft and global adjustment across all items, simultaneously increasing the estimated propensities of tail items and moderating those of popular items. We empirically compare Mult-CIPS and Mult-BiW in Section~\ref{backbone}.

\subsubsection{Applicability of Doubly Robust Estimation}

Another important class of debiasing techniques is \textit{doubly robust} (DR) estimation, which combines propensity-based reweighting with outcome imputation to achieve unbiasedness even when one of the two components is misspecified. DR estimators have been extensively studied in the context of explicit feedback recommendation \cite{wang2019doubly,song2023cdr}. In these settings, the propensity typically corresponds to the probability that a user’s rating is observed, i.e., $P(o_{ui}=1)$, where $o_{ui}$ indicates a rating event. This formulation allows explicit modeling of both exposure and user preference signals (e.g., high ratings indicating relevance).

In contrast, our work focuses on implicit feedback scenarios, where only binary interaction signals are available and propensities are defined as $P(y_{ui}=1 \!\mid\! r_{ui}=1)$. The absence of explicit observation indicators $o_{ui}$ makes the direct application of DR estimators challenging. Moreover, extending Mult-BiW to a DR framework would require carefully designed strategies for imputing relevance probabilities $\gamma_{ui}^{\mathrm{imputation}}$ and interaction outcomes $y_{ui}^{\mathrm{imputation}}$. While promising, such extensions fall outside the scope of this paper and are left for future work.

\subsubsection{Model Complexity}
The proposed Mult-IPS and Mult-BiW methods are computationally efficient. The propensity weights are precomputed from item popularity statistics and reused throughout training, incurring only $\mathcal{O}(|\mathcal{I}|)$ additional storage, where $|\mathcal{I}|$ is the number of items.

The dominant computational cost arises from optimizing the multinomial likelihood. Specifically, the per-epoch training complexity of the proposed methods is $\mathcal{O}(|\mathcal{U}||\mathcal{I}|d)$, where $|\mathcal{U}|$ is the number of users and $d$ is the embedding dimension, as the likelihood is evaluated over the full item set for each user. In contrast, pointwise or pairwise objectives with negative sampling scale with the number of observed interactions $|\mathcal{D}|$, requiring $\mathcal{O}(|\mathcal{D}|(1\!+\!K)d)$ operations per epoch, where $K$ is the number of negative samples.

While pointwise or pairwise losses are asymptotically more efficient when $|\mathcal{D}|(1\!+\!K) \ll |\mathcal{U}||\mathcal{I}|$, the empirical runtime gap is often modest in practice due to modern hardware acceleration. In particular, the multinomial loss can be implemented using dense matrix multiplication followed by a softmax operation, both of which are highly optimized in contemporary GPU libraries and benefit substantially from large-batch computation. To further substantiate the efficiency of Mult-BiW, we report empirical comparisons of training time against representative baselines on all evaluated datasets in Section~\ref{efficiency}.

\section{Experiments}
\label{sec:exp}

\begin{table}[!t]
    \centering
    \caption{Summary of the evaluation datasets. The ML10M and Netflix datasets are adopted from \cite{zheng2021disentangling}.}
    \begin{tabular}{lcccccc}
        \toprule
        Statistics & Amazon & Google & Gowalla & Yelp & ML10M & Netflix \\
        \midrule
        \# Users & 39,240 & 28,943 & 43,457 & 36,948 & 37,962 & 32,450 \\
        \# Items & 18,536 & 8,035 & 15,294 & 20,950 & 4,819 & 8,432 \\
        \# Interactions & 882,123 & 306,807 & 727,163 & 1,299,619 & 1,371,473 & 2,212,690\\
        Density & 0.1213\% & 0.1319\% & 0.1094\% & 0.1679\% & 0.7497\% & 0.8087\% \\
        $\max p_{i}^\textrm{train}$ & 1,170 & 360 & 2,304 & 1,410 & 11,739 & 8,763 \\
        $\min p_{i}^\textrm{train}$ & 5 & 8 & 5 & 2 & 0 & 0 \\
        $p_{i}^\textrm{valid}$ & 5 & 4 & 5 & 6 & - & - \\
        $p_{i}^\textrm{test}$ & 10 & 8 & 10 & 12 & - & - \\
        \bottomrule
    \end{tabular}
    \label{tab:dataset}
\end{table}

In this section, we conduct extensive empirical studies on six real-world datasets to evaluate the effectiveness of the proposed Mult-BiW framework. The implementation and preprocessed datasets are publicly available on GitHub.\footnote{\url{https://github.com/zhuty16/Mult-BiW}.} The experiments are designed to answer the following research questions:

\begin{itemize}
    \item[\textbf{RQ1.}] How does Mult-BiW perform compared with state-of-the-art popularity debiasing methods?
    \item[\textbf{RQ2.}] What is the contribution of each component in Mult-BiW?
    \item[\textbf{RQ3.}] How sensitive is Mult-BiW to key hyperparameters?
    \item[\textbf{RQ4.}] Can Mult-BiW effectively mitigate popularity bias in recommendation policies?
    \item[\textbf{RQ5.}] How does Mult-BiW perform across item groups with varying popularity levels?
\end{itemize}

\subsection{Experimental Settings}

\subsubsection{Datasets}
We evaluate all methods on six widely used real-world datasets. The first four datasets are Amazon Books \cite{he2016ups},\footnote{\url{http://jmcauley.ucsd.edu/data/amazon/index\_2014.html}.} Google Local \cite{he2017translation},\footnote{\url{http://cseweb.ucsd.edu/~jmcauley/datasets.html}.} Gowalla \cite{cho2011friendship},\footnote{\url{http://snap.stanford.edu/data/loc-Gowalla.html}.} and Yelp Challenge 2019.\footnote{\url{https://www.yelp.com/dataset/}.} Following common practice \cite{he2020lightgcn}, we applied a $20$-core preprocessing strategy, retaining users and items with at least $20$ interactions.

To evaluate popularity debiasing under controlled conditions, we adopted the offline evaluation protocol proposed in prior work \cite{liang2016causal}. Specifically, we constructed a \emph{balanced test set} in which items appear with equal frequency, mimicking the outcome of a uniformly random recommendation policy. This evaluation setup is consistent with that used in causal and debiased recommendation studies \cite{liang2016causal,bonner2018causal,zheng2021disentangling,wei2021model}. A balanced validation set was constructed in the same manner, while the remaining interactions formed the imbalanced training set. The data were split into training, validation, and test sets with a ratio of $70\%/10\%/20\%$ \cite{zheng2021disentangling}. As a result, models were trained on long-tailed (biased) data and evaluated on balanced (unbiased) data.

In addition, we included two public preprocessed datasets, i.e., MovieLens-10M (ML10M) and Netflix, released with DICE~\cite{zheng2021disentangling},\footnote{\url{https://github.com/tsinghua-fib-lab/DICE/tree/main/data}.} which are identical to those used for evaluating DICE. Unlike the other datasets, the validation and test sets of ML10M and Netflix are not perfectly balanced, as they were constructed using random sampling with probability capping. Summary statistics of all datasets are reported in Table~\ref{tab:dataset}.

Finally, we also conducted experiments on the Coat \cite{schnabel2016recommendations} and Yahoo! R3 \cite{marlin2009collaborative} datasets, which naturally contain uniformly distributed test sets. Due to space constraints, these results are reported in the Appendix.

\begin{table*}[!t]
    \small
    \centering
    \caption{Performance comparison of the evaluated methods using the MF backbone across six datasets. Boldface indicates the best results, and underlining marks the second-best. Asterisk (*) denotes statistically significant improvements over the best baseline based on a two-sample t-test ($p<0.01$) with five repeated runs.}
    \begin{tabular}{llccccccccc}
        \toprule
        Dataset & Metric & Expo-MF & Rel-MF & UBPR & DICE & MACR & CPR & TTEN & Mult-BiW & Improv. \\
        \midrule
        \multirow{4}{*}{Amazon} 
            & R@10 & 0.0445 & 0.0470 & 0.0552 & 0.0516 & 0.0536 & \underline{0.0637} & 0.0635 & \textbf{0.0807}* & 26.60\% \\
            & N@10 & 0.0322 & 0.0363 & 0.0457 & 0.0433 & 0.0451 & 0.0544 & \underline{0.0559} & \textbf{0.0716}* & 28.19\% \\
            & R@20 & 0.0709 & 0.0742 & 0.0866 & 0.0807 & 0.0822 & \underline{0.0974} & 0.0929 & \textbf{0.1161}* & 19.21\% \\
            & N@20 & 0.0418 & 0.0460 & 0.0563 & 0.0531 & 0.0549 & \underline{0.0657} & 0.0656 & \textbf{0.0834}* & 26.96\% \\
        \midrule
        \multirow{4}{*}{Google} 
            & R@10 & 0.0765 & 0.1264 & 0.1314 & 0.1278 & 0.1299 & \textbf{0.1471} & 0.1370 & \underline{0.1441} & -1.99\% \\
            & N@10 & 0.0494 & 0.0821 & 0.0896 & 0.0869 & 0.0869 & \underline{0.0999} & 0.0965 & \textbf{0.1016} & 1.75\% \\
            & R@20 & 0.1161 & 0.1933 & 0.1980 & 0.1937 & 0.1958 & \textbf{0.2205} & 0.2020 & \underline{0.2122} & -3.76\% \\
            & N@20 & 0.0623 & 0.1036 & 0.1109 & 0.1079 & 0.1082 & \textbf{0.1235} & 0.1174 & \underline{0.1233} & -0.14\% \\
        \midrule
        \multirow{4}{*}{Gowalla} 
            & R@10 & 0.0712 & 0.0882 & 0.1011 & 0.0932 & 0.1040 & \underline{0.1100} & 0.1098 & \textbf{0.1288}* & 17.11\% \\
            & N@10 & 0.0450 & 0.0560 & 0.0705 & 0.0648 & 0.0727 & 0.0783 & \underline{0.0796} & \textbf{0.0944}* & 18.62\% \\
            & R@20 & 0.1165 & 0.1424 & 0.1626 & 0.1529 & 0.1640 & \underline{0.1782} & 0.1690 & \textbf{0.1970}* & 10.56\% \\
            & N@20 & 0.0600 & 0.0741 & 0.0901 & 0.0840 & 0.0919 & \underline{0.0998} & 0.0984 & \textbf{0.1161}* & 16.38\% \\
        \midrule
        \multirow{4}{*}{Yelp} 
            & R@10 & 0.0211 & 0.0165 & 0.0246 & 0.0233 & 0.0243 & \underline{0.0291} & \underline{0.0291} & \textbf{0.0359}* & 23.26\% \\
            & N@10 & 0.0164 & 0.0143 & 0.0230 & 0.0216 & 0.0226 & 0.0268 & \underline{0.0273} & \textbf{0.0342}* & 25.52\% \\
            & R@20 & 0.0377 & 0.0310 & 0.0443 & 0.0416 & 0.0439 & \underline{0.0519} & 0.0506 & \textbf{0.0632}* & 21.88\% \\
            & N@20 & 0.0230 & 0.0197 & 0.0299 & 0.0280 & 0.0296 & \underline{0.0347} & 0.0346 & \textbf{0.0435}* & 25.19\% \\
        \midrule
        \multirow{4}{*}{ML10M} 
            & R@10 & 0.0965 & 0.0698 & 0.1010 & 0.1124 & 0.0959 & \underline{0.1249} & 0.1204 & \textbf{0.1287}* & 3.04\% \\
            & N@10 & 0.0756 & 0.0608 & 0.0858 & 0.0996 & 0.0846 & \underline{0.1109} & 0.1068 & \textbf{0.1183}* & 6.66\% \\
            & R@20 & 0.1456 & 0.1116 & 0.1612 & 0.1807 & 0.1562 & \underline{0.1972} & 0.1878 & \textbf{0.1979} & 0.36\% \\
            & N@20 & 0.0941 & 0.0737 & 0.1056 & 0.1214 & 0.1046 & \underline{0.1342} & 0.1280 & \textbf{0.1393}* & 3.77\% \\
        \midrule
        \multirow{4}{*}{Netflix} 
            & R@10 & 0.0877 & 0.0515 & 0.0830 & 0.0884 & 0.0798 & \underline{0.0985} & 0.0960 & \textbf{0.1118}* & 13.47\% \\
            & N@10 & 0.0736 & 0.0599 & 0.0981 & 0.1089 & 0.0944 & \underline{0.1208} & 0.1169 & \textbf{0.1461}* & 20.97\% \\
            & R@20 & 0.1228 & 0.0832 & 0.1292 & 0.1371 & 0.1266 & \underline{0.1522} & 0.1466 & \textbf{0.1649}* & 8.35\% \\
            & N@20 & 0.0873 & 0.0675 & 0.1090 & 0.1191 & 0.1057 & \underline{0.1319} & 0.1273 & \textbf{0.1539}* & 16.64\% \\
        \bottomrule
    \end{tabular}
    \label{tab:result1}
\end{table*}

\begin{table}[!t]
    \centering
    \caption{Training time comparison between Mult-BiW-MF and CPR-MF across six datasets. Boldface indicates the best results. Here, "s" and "m" denote seconds and minutes, respectively.}
    \begin{tabular}{lcccccc}
        \toprule
            \multirow{3}{*}{Dataset} &  \multicolumn{3}{c}{Mult-BiW-MF} & \multicolumn{3}{c}{CPR-MF} \\
            \cmidrule(lr){2-4} \cmidrule(lr){5-7} 
            & Per Epoch & \# Epochs & Total & Per Epoch & \# Epochs & Total \\
        \midrule
            Amazon & 3.76 s & 250 & \textbf{15.67 m} & 2.09 s & 592 & 20.62 m \\
            Google & 1.04 s & 240 & \textbf{4.15 m} & 0.63 s & 608 & 6.34 m \\
            Gowalla & 2.77 s & 250 & \textbf{11.53 m} & 1.19 s & 700 & 13.89 m \\
            Yelp & 3.93 s & 270 & \textbf{17.71 m} & 2.87 s & 500 & 23.89 m \\
            ML10M & 2.59 s & 110 & 4.75 m & 2.54 s & 96 & \textbf{4.06 m} \\
            Netflix & 4.07 s & 140 & 9.50 m & 3.98 s & 84 & \textbf{5.57 m} \\
        \bottomrule
    \end{tabular}
    \label{tab:training_time}
\end{table}

\subsubsection{Baselines}
We compare Mult-BiW with the following state-of-the-art popularity debiasing methods:
\begin{itemize}
    \item \textbf{Expo-MF} \cite{liang2016modeling}: This method models user exposure as a latent variable and infers it jointly with preferences in collaborative filtering.
    \item \textbf{Rel-MF} \cite{saito2020unbiased}: This is an IPS-based method using a pointwise unbiased estimator to optimize user-item relevance.
    \item \textbf{UBPR} \cite{saito2020unbiased2}: This is an IPS-based method employing an unbiased pairwise ranking loss to learn user-item relevance.
    \item \textbf{DICE} \cite{zheng2021disentangling}: This method learns disentangled embeddings for user interest and conformity by sampling cause-specific data.
    \item \textbf{MACR} \cite{wei2021model}: This method uses multi-task learning to decouple user-item matching, user conformity, and item popularity.
    \item \textbf{CPR} \cite{wan2022cross}: This method constructs unbiased losses by combining predictions from multiple observed interactions.
    \item \textbf{TTEN} \cite{kim2023test}: This method controls item popularity effects by normalizing item embeddings at inference time.
\end{itemize}

We further include two variants of the proposed method:
\begin{itemize}
    \item \textbf{Mult-IPS}: This method combines the multinomial likelihood with vanilla IPS for popularity debiasing.
    \item \textbf{Mult-CIPS}: This method applies propensity clipping to Mult-IPS to reduce propensity estimator variance.
\end{itemize}

Most debiasing methods can be paired with different backbone models. For a fair comparison, we adopted matrix factorization (MF) \cite{koren2009matrix} as the default backbone across all methods. In Section~\ref{backbone}, we further evaluate IPS-based methods using the state-of-the-art LightGCN backbone \cite{he2020lightgcn}. Additional results involving recent baselines \cite{zhang2023invariant,zhang2024robust,liu2023popularity,bonner2018causal} are included in the Appendix.

\subsubsection{Evaluation Metrics}
We evaluated top-$N$ recommendation performance using two widely adopted metrics: Recall@$N$ (R@$N$) and NDCG@$N$ (N@$N$) \cite{wang2019neural,he2020lightgcn,zheng2021disentangling}. By default, $N$ was set to $10$ and $20$. Since items are uniformly distributed in the test sets, higher values of these metrics indicate improvements in both accuracy and fairness. Following standard practice \cite{krichene2020sampled}, for each user, all items with which the user had not interacted during training or validation were treated as candidates and ranked. Each model was randomly initialized five times, and the average performance was reported.

\subsubsection{Hyperparameter Settings} 

Our proposed framework was implemented in TensorFlow. For baseline models, we used the official code provided in their respective repositories. For all models, the embedding dimension was fixed at $d=64$. The learning rate was tuned in \{1e-3, 5e-4, 2e-4, 1e-4\}, and the $l_2$ regularization coefficient was selected from \{0, 1e-7, 1e-6, 1e-5, 1e-4\}. Models were trained using the Adam optimizer \cite{kingma2014adam} with a batch size of $256$.

For Mult-BiW with a fixed $\alpha$ (Mult-FBiW), $\alpha$ was tuned in $\{0, 0.1, \ldots, 1.0\}$. For Progressive Bi-Weighting (Mult-PBiW), the cumulative learning parameter $\eta$ was tuned in $\{0.5, 1.0, 2.0\}$. Hyperparameters of all baseline methods were tuned on the validation set following the recommendations in their respective papers, and the best-performing configurations were reported.

\subsection{Main Results (RQ1)}

\begin{table*}[!t]
    \small
    \centering
    \caption{Performance comparison with IPS-based methods using the Mult-MF and Mult-LightGCN backbones across six datasets. For each backbone, the best results are highlighted in boldface.}
    \begin{tabular}{llcccccccccc}
        \toprule
            \multirow{3}{*}{Dataset} & \multirow{3}{*}{Metric} &  \multicolumn{5}{c}{Mult-MF} & \multicolumn{5}{c}{Mult-LightGCN} \\
            \cmidrule(lr){3-7} \cmidrule(lr){8-12} 
            & & None & IPS & CIPS & FBiW & PBiW & None & IPS & CIPS & FBiW & PBiW \\
        \midrule
        \multirow{4}{*}{Amazon} 
            & R@10 & 0.0631 & 0.0661 & 0.0738 & 0.0775 & \textbf{0.0807} & 0.0694 & 0.0882 & 0.0898 & 0.0921 & \textbf{0.0980} \\
            & N@10 & 0.0544 & 0.0586 & 0.0667 & 0.0691 & \textbf{0.0716} & 0.0616 & 0.0804 & 0.0813 & 0.0841 & \textbf{0.0890} \\
            & R@20 & 0.0947 & 0.0944 & 0.1069 & 0.1115 & \textbf{0.1161} & 0.1013 & 0.1256 & 0.1269 & 0.1311 & \textbf{0.1396} \\
            & N@20 & 0.0651 & 0.0680 & 0.0776 & 0.0803 & \textbf{0.0834} & 0.0720 & 0.0926 & 0.0933 & 0.0967 & \textbf{0.1024} \\
        \midrule
        \multirow{4}{*}{Google} 
            & R@10 & 0.1293 & 0.1300 & 0.1393 & 0.1392 & \textbf{0.1441} & 0.1512 & 0.1594 & 0.1604 & \textbf{0.1650} & 0.1622 \\
            & N@10 & 0.0896 & 0.0916 & 0.0970 & 0.0985 & \textbf{0.1016} & 0.1062 & 0.1098 & 0.1112 & \textbf{0.1151} & 0.1148 \\
            & R@20 & 0.1901 & 0.1904 & 0.2041 & 0.2042 & \textbf{0.2122} & 0.2180 & 0.2289 & 0.2298 & 0.2349 & \textbf{0.2355} \\
            & N@20 & 0.1091 & 0.1112 & 0.1177 & 0.1190 & \textbf{0.1233} & 0.1276 & 0.1324 & 0.1336 & 0.1375 & \textbf{0.1384} \\
        \midrule
        \multirow{4}{*}{Gowalla} 
            & R@10 & 0.1062 & 0.1095 & 0.1214 & 0.1244 & \textbf{0.1288} & 0.1062 & 0.1390 & 0.1415 & 0.1428 & \textbf{0.1498} \\
            & N@10 & 0.0771 & 0.0802 & 0.0894 & 0.0922 & \textbf{0.0944} & 0.0794 & 0.1022 & 0.1055 & 0.1077 & \textbf{0.1128} \\
            & R@20 & 0.1683 & 0.1664 & 0.1845 & 0.1896 & \textbf{0.1970} & 0.1649 & 0.2098 & 0.2125 & 0.2144 & \textbf{0.2255} \\
            & N@20 & 0.0968 & 0.0985 & 0.1096 & 0.1128 & \textbf{0.1161} & 0.0976 & 0.1249 & 0.1279 & 0.1303 & \textbf{0.1367} \\
        \midrule
        \multirow{4}{*}{Yelp} 
            & R@10 & 0.0255 & 0.0259 & 0.0311 & 0.0354 & \textbf{0.0359} & 0.0280 & 0.0354 & 0.0387 & 0.0418 & \textbf{0.0455} \\
            & N@10 & 0.0238 & 0.0237 & 0.0297 & 0.0338 & \textbf{0.0342} & 0.0274 & 0.0332 & 0.0365 & 0.0404 & \textbf{0.0434} \\
            & R@20 & 0.0462 & 0.0454 & 0.0556 & 0.0612 & \textbf{0.0632} & 0.0480 & 0.0631 & 0.0674 & 0.0723 & \textbf{0.0784} \\
            & N@20 & 0.0310 & 0.0306 & 0.0381 & 0.0425 & \textbf{0.0435} & 0.0339 & 0.0428 & 0.0464 & 0.0507 & \textbf{0.0545} \\
        \midrule
        \multirow{4}{*}{ML10M} 
            & R@10 & 0.1132 & 0.0817 & 0.1023 & 0.1258 & \textbf{0.1287} & 0.1302 & 0.0992 & 0.1084 & 0.1333 & \textbf{0.1354} \\
            & N@10 & 0.1011 & 0.0691 & 0.0930 & 0.1162 & \textbf{0.1183} & 0.1176 & 0.0861 & 0.0988 & 0.1227 & \textbf{0.1240} \\
            & R@20 & 0.1748 & 0.1352 & 0.1639 & 0.1937 & \textbf{0.1979} & 0.2031 & 0.1595 & 0.1749 & 0.2060 & \textbf{0.2103} \\
            & N@20 & 0.1202 & 0.0869 & 0.1122 & 0.1369 & \textbf{0.1393} & 0.1403 & 0.1055 & 0.1197 & 0.1450 & \textbf{0.1473} \\
        \midrule
        \multirow{4}{*}{Netflix} 
            & R@10 & 0.0937 & 0.0641 & 0.0897 & 0.1096 & \textbf{0.1118} & 0.1045 & 0.0834 & 0.0963 & 0.1138 & \textbf{0.1143} \\
            & N@10 & 0.1144 & 0.0796 & 0.1209 & 0.1452 & \textbf{0.1461} & 0.1320 & 0.1038 & 0.1274 & 0.1470 & \textbf{0.1480} \\
            & R@20 & 0.1392 & 0.0995 & 0.1364 & 0.1617 & \textbf{0.1649} & 0.1557 & 0.1285 & 0.1444 & 0.1695 & \textbf{0.1719} \\
            & N@20 & 0.1231 & 0.0870 & 0.1274 & 0.1522 & \textbf{0.1539} & 0.1408 & 0.1130 & 0.1348 & 0.1562 & \textbf{0.1577} \\
        \bottomrule
    \end{tabular}
    \label{tab:result2}
\end{table*}

\subsubsection{Effectiveness Comparison}

Table~\ref{tab:result1} reports the top-$N$ recommendation performance of all compared methods using MF as the backbone model. Several observations can be made. First, among IPS-based baselines, the pairwise unbiased learning method UBPR consistently outperforms the pointwise method Rel-MF, indicating that pairwise objectives are more effective than pointwise ones. Second, CPR and TTEN achieve the strongest performance among all baseline methods, reflecting the benefits of cross-interaction modeling and inference-time normalization, respectively.

Across all datasets, the proposed Mult-BiW consistently achieves the best performance. Compared with the strongest baseline on each dataset, Mult-BiW yields an average improvement of $9.43\%$ in Recall@20 and $14.80\%$ in NDCG@20. These gains are particularly pronounced on datasets with more severe popularity skew (e.g., Amazon Books and Yelp), suggesting that Mult-BiW is especially effective under strong long-tail distributions.

Notably, Rel-MF and UBPR are also IPS-based methods, but rely on pointwise logistic loss and pairwise ranking loss, respectively. The superior performance of Mult-BiW highlights the advantages of adopting a listwise objective based on the multinomial likelihood, which captures global user preferences over the entire item set, and introducing Bi-Weighting to smooth propensity estimation and mitigate estimation errors for tail items.

\subsubsection{Efficiency Comparison}
\label{efficiency}
Table~\ref{tab:training_time} compares the training efficiency of Mult-BiW and CPR under the MF backbone. As reported in \cite{wan2022cross}, CPR achieves high efficiency via dynamic interaction sampling, making it substantially faster than several earlier debiasing methods such as Rel-MF, UBPR, and DICE. Despite this advantage, Mult-BiW exhibits comparable training time on ML10M and Netflix, and is even faster than CPR on the remaining four datasets.

This efficiency stems from the fact that Mult-BiW introduces no additional model parameters or sampling procedures. The Bi-Weighting factors are precomputed from item popularity statistics and incur negligible overhead during training. These results demonstrate that Mult-BiW not only improves recommendation quality but also maintains high computational efficiency, making it suitable for large-scale applications.

\begin{figure}[!t]
    \centering
    \subfloat{\includegraphics[width=0.33\linewidth]{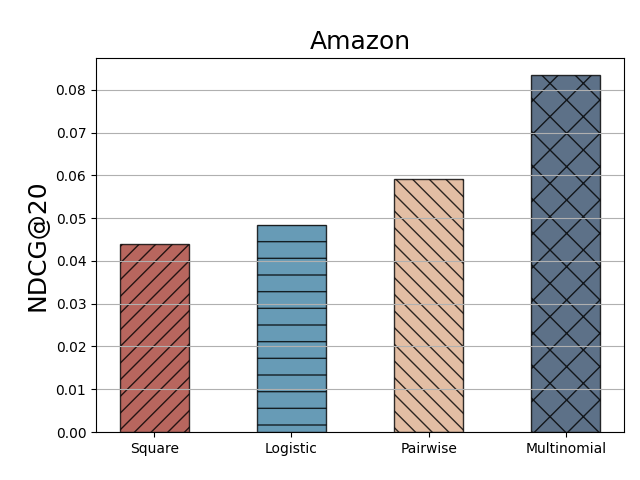}}
    \hfil
    \subfloat{\includegraphics[width=0.33\linewidth]{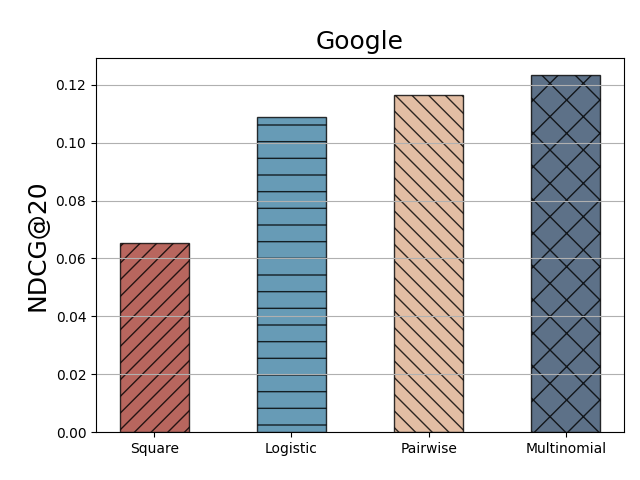}}
    \hfil
    \subfloat{\includegraphics[width=0.33\linewidth]{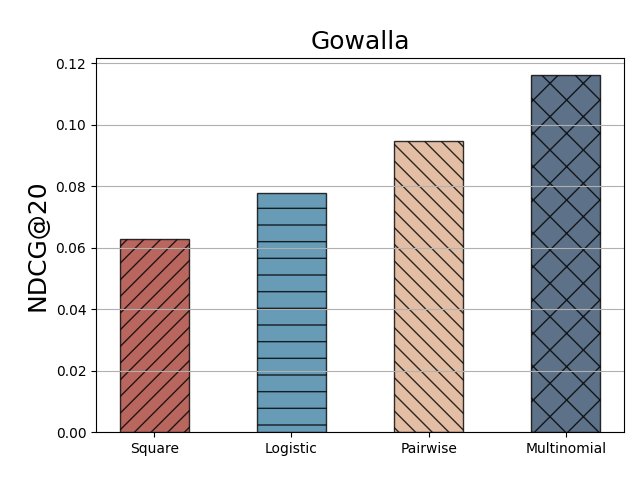}}
    \hfil
    \subfloat{\includegraphics[width=0.33\linewidth]{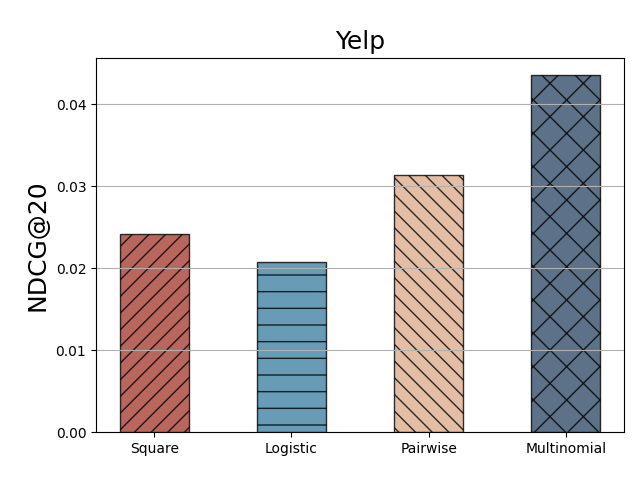}}
    \hfil
    \subfloat{\includegraphics[width=0.33\linewidth]{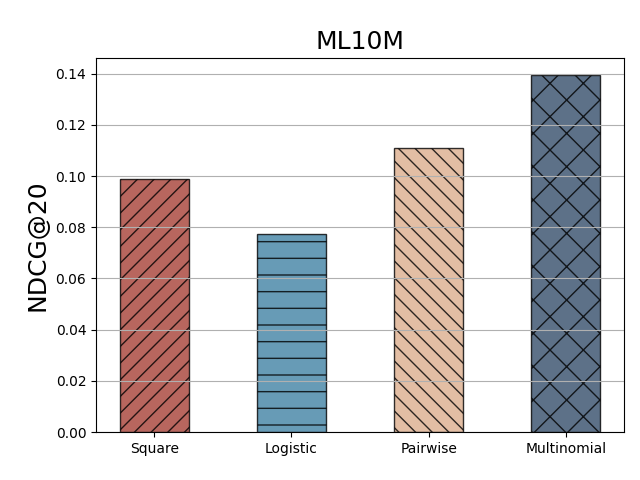}}
    \hfil
    \subfloat{\includegraphics[width=0.33\linewidth]{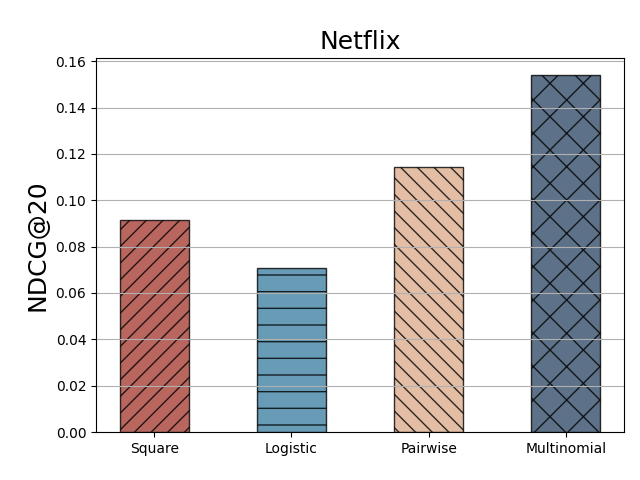}}
    \caption{Performance of BiW-MF with different loss functions on six datasets.}
    \label{fig:loss}
\end{figure}

\begin{figure}[!t]
    \centering
    \subfloat{\includegraphics[width=0.33\linewidth]{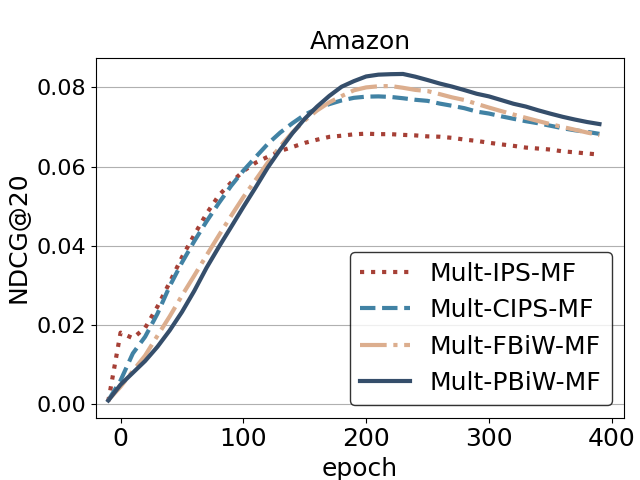}}
    \hfil
    \subfloat{\includegraphics[width=0.33\linewidth]{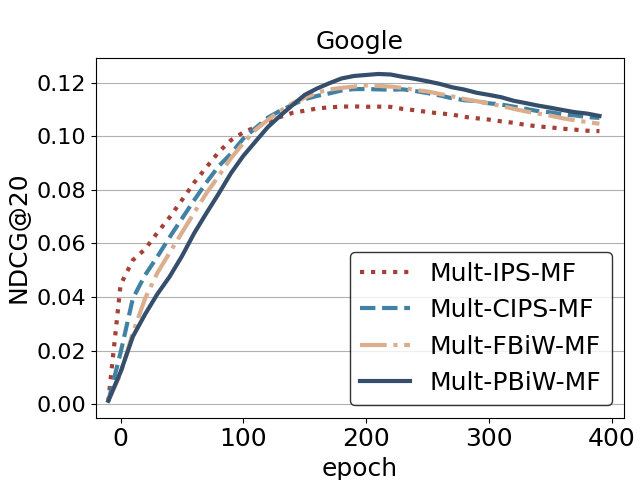}}
    \hfil
    \subfloat{\includegraphics[width=0.33\linewidth]{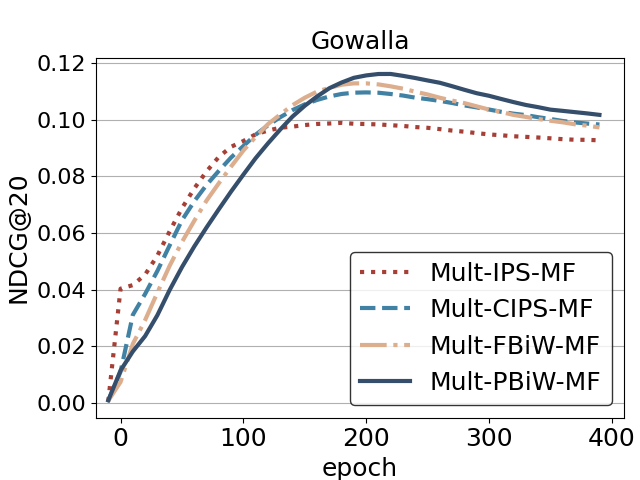}}
    \hfil
    \subfloat{\includegraphics[width=0.33\linewidth]{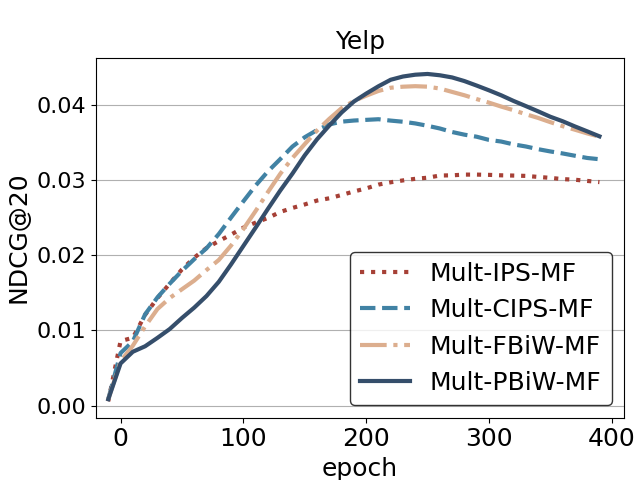}}
    \hfil
    \subfloat{\includegraphics[width=0.33\linewidth]{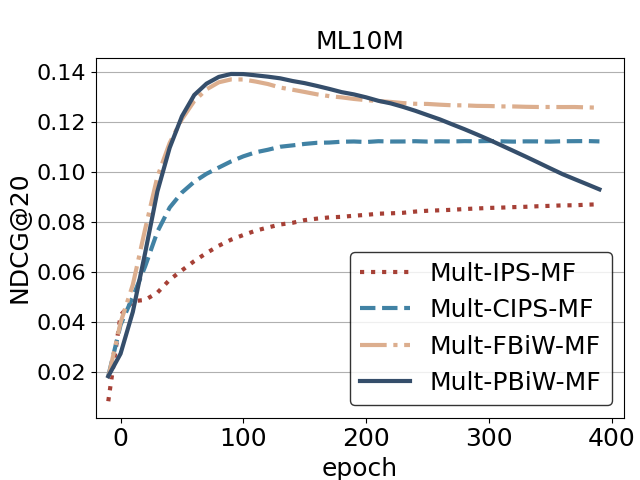}}
    \hfil
    \subfloat{\includegraphics[width=0.33\linewidth]{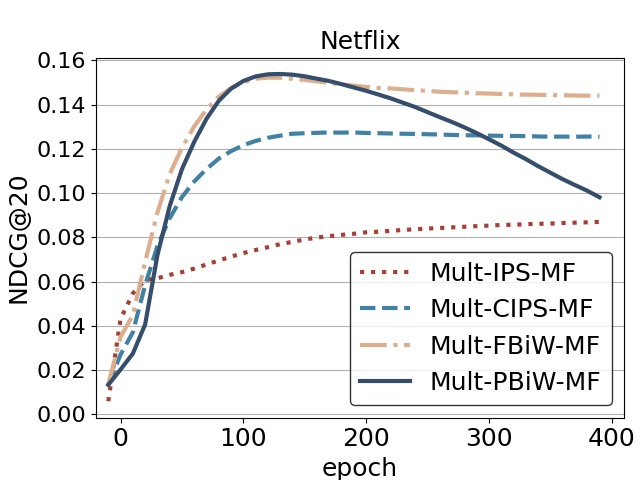}}
    \caption{Performance curves of Mult-IPS-MF, Mult-CIPS-MF, Mult-FBiW-MF, and Mult-PBiW-MF on six datasets.}
    \label{fig:curve}
\end{figure}

\subsection{Ablation Studies (RQ2)}

\subsubsection{Loss Function Comparison} 

Figure~\ref{fig:loss} compares the performance of BiW-MF under different loss functions. The multinomial likelihood consistently outperforms all alternatives across datasets, confirming the effectiveness of listwise unbiased learning for recommendation. By modeling the full item distribution for each user, the multinomial likelihood enables the model to learn global preference structures rather than relying on local comparisons.

The pairwise loss achieves the second-best performance, benefiting from its ranking-aware nature. In contrast, pointwise squared loss and pointwise logistic loss perform substantially worse, likely due to their limited capacity to capture relative item importance and their sensitivity to popularity-induced label imbalance. These results validate our design choice of combining IPS with a listwise objective.

\subsubsection{Bi-Weighting vs. IPS}
\label{backbone}

Table~\ref{tab:result2} presents the performance of IPS-based methods under two backbone models: Mult-MF and Mult-LightGCN. The vanilla IPS and clipped IPS (CIPS) methods exhibit unstable and often suboptimal performance, particularly on ML10M and Netflix. This behavior can be attributed to inaccurate propensity estimation and the large variance induced by extreme inverse weights, which are only partially mitigated by clipping.

In contrast, the proposed Bi-Weighting strategy substantially improves robustness by smoothing the estimated propensities. Both fixed Bi-Weighting (FBiW) and progressive Bi-Weighting (PBiW) consistently outperform IPS and CIPS across datasets and backbone models. Moreover, PBiW achieves the best overall performance in most cases, demonstrating the benefit of progressively shifting from representation learning to popularity debiasing during training.

Figure~\ref{fig:curve} further illustrates the test performance of IPS, CIPS, FBiW, and PBiW over training epochs using MF as the backbone. IPS and CIPS show either slow convergence or early overfitting, whereas Bi-Weighting methods achieve faster and more stable improvements. PBiW consistently outperforms FBiW, highlighting the effectiveness of cumulative learning. However, the figure also reveals that overly aggressive debiasing (i.e., continuously decreasing $\alpha$) may lead to overfitting, underscoring the importance of controlled smoothing in propensity estimation.

\begin{figure}[!t]
    \centering
    \subfloat{\includegraphics[width=0.33\linewidth]{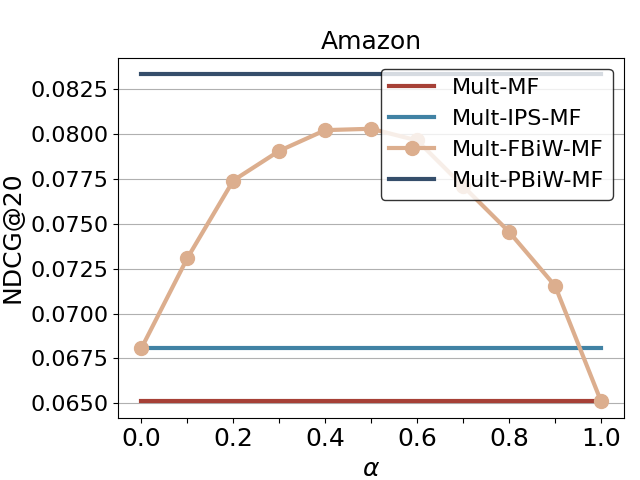}}
    \hfil
    \subfloat{\includegraphics[width=0.33\linewidth]{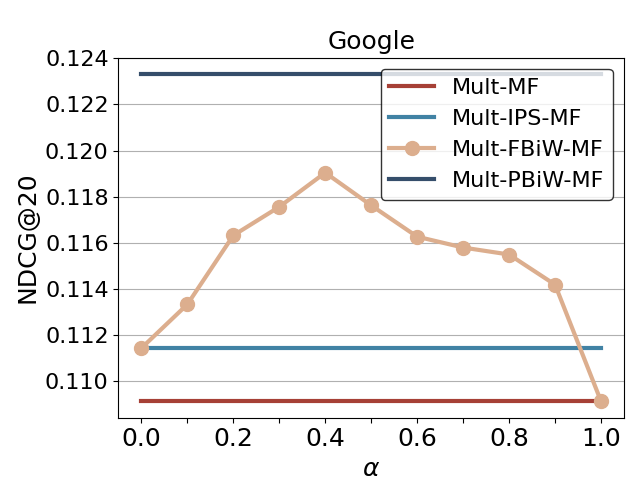}}
    \hfil
    \subfloat{\includegraphics[width=0.33\linewidth]{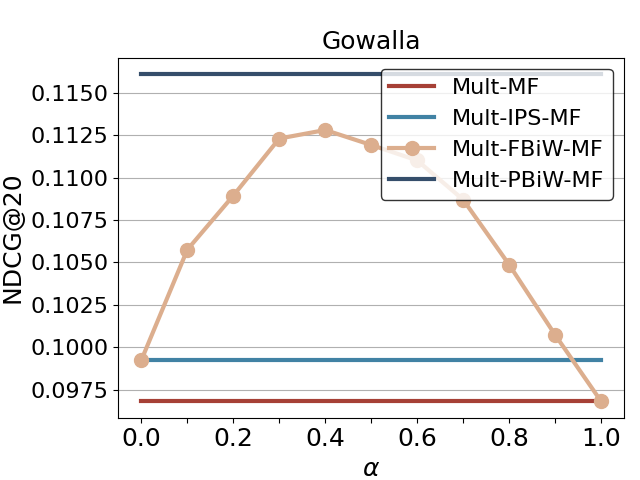}}
    \hfil
    \subfloat{\includegraphics[width=0.33\linewidth]{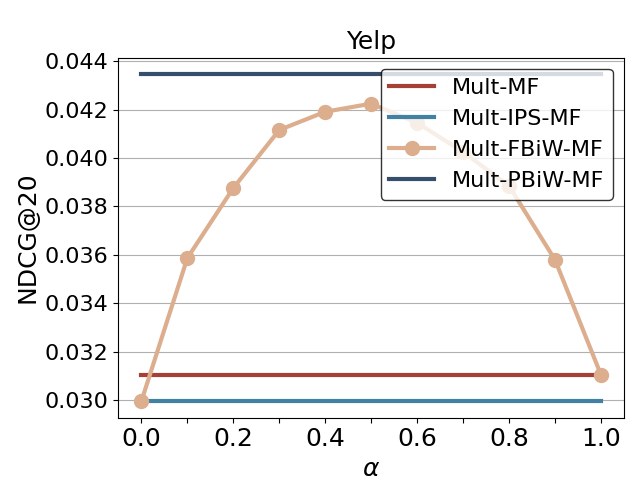}}
    \hfil
    \subfloat{\includegraphics[width=0.33\linewidth]{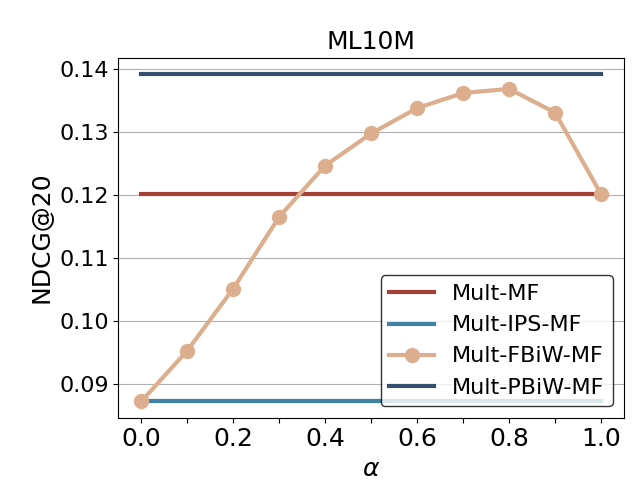}}
    \hfil
    \subfloat{\includegraphics[width=0.33\linewidth]{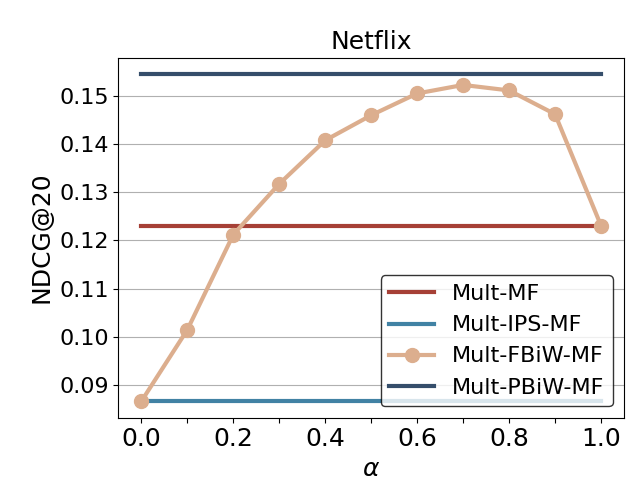}}
    \caption{Performance of Mult-MF, Mult-IPS-MF, Mult-FBiW-MF, and Mult-PBiW-MF  w.r.t. $\alpha$ on six datasets.}
    \label{fig:alpha}
\end{figure}

\subsection{Hyperparameter Studies (RQ3)}
\subsubsection{Effect of Model Balance Coefficient $\alpha$}

The performance of Bi-Weighting critically depends on the balance coefficient $\alpha$, which controls the trade-off between individual propensities and the collection model. Using Mult-MF as the backbone, we vary $\alpha$ in $\{0, 0.1, \ldots, 1.0\}$, with results shown in Figure~\ref{fig:alpha}. When $\alpha=0$, FBiW reduces to vanilla IPS, while $\alpha=1$ corresponds to the naive estimator that ignores propensities entirely.

The results show that FBiW significantly outperforms IPS for a wide range of intermediate $\alpha$ values, confirming that appropriate smoothing improves the robustness of propensity estimation. Furthermore, PBiW consistently surpasses FBiW, even when FBiW uses its optimal $\alpha$, demonstrating the advantage of dynamically balancing representation learning and debiasing throughout training.

\subsubsection{Effect of Cumulative Learning Parameter $\eta$}
\label{sec:eta}

We further analyze the impact of the cumulative learning parameter $\eta$, which controls the decay rate of $\alpha$ in PBiW. Figure~\ref{fig:eta} reports results for $\eta \in \{0.5, 1.0, 2.0\}$ using Mult-MF and Mult-LightGCN backbones. For MF, smaller values of $\eta$ generally yield better performance, suggesting that MF benefits from earlier emphasis on debiasing. In contrast, LightGCN performs better with larger $\eta$ values, which allow more epochs for representation learning before stronger debiasing is applied.

This difference can be explained by the increased complexity of representation learning in LightGCN, which relies on multilayer graph convolution to propagate embeddings. A slower transition to debiasing helps preserve representation quality before correcting popularity bias.

\subsection{Debiasing Visualization (RQ4)}

Figure~\ref{fig:visualize} visualizes the effectiveness of different methods in mitigating popularity bias by plotting the relationship between item popularity in the training data and item frequency in top-$10$ recommendation lists. Results are shown for Mult-MF, CPR-MF, UBPR-MF, and Mult-BiW-MF on three representative datasets.

Without debiasing, Mult-MF exhibits a strong positive correlation between training popularity and recommendation frequency, indicating a heavy bias toward popular items. All debiasing methods reduce this correlation, demonstrating their ability to control popularity bias. However, the magnitude of debiasing differs substantially. Compared with CPR-MF and UBPR-MF, Mult-BiW-MF produces the flattest curves and the lowest Pearson correlation coefficients, indicating a stronger and more consistent debiasing effect. These results confirm that Mult-BiW-MF effectively promotes more balanced item exposure in recommendation policies.

\subsection{Heterogeneous Analysis of Item Groups (RQ5)}

To investigate how debiasing affects different segments of the item distribution, we analyze the performance of Mult-MF, Mult-IPS-MF, and Mult-BiW-MF across item groups partitioned by popularity. The average NDCG@10 for each group is shown in Figure~\ref{fig:group}.

The results indicate that the performance gains achieved by Mult-IPS-MF and Mult-BiW-MF over Mult-MF stem from an explicit accuracy trade-off across item groups: while recommendation accuracy on head items is moderately reduced, substantial improvements are observed for tail items. This shift reflects the effectiveness of IPS-based debiasing in mitigating popularity bias. Moreover, Mult-BiW-MF consistently outperforms Mult-IPS-MF, with particularly pronounced gains on less popular items. These findings highlight the advantage of the proposed Bi-Weighting strategy in improving the robustness of propensity estimation and achieving a more balanced performance across heterogeneous item groups.

\begin{figure}[!t]
    \centering
    \subfloat{\includegraphics[width=0.33\linewidth]{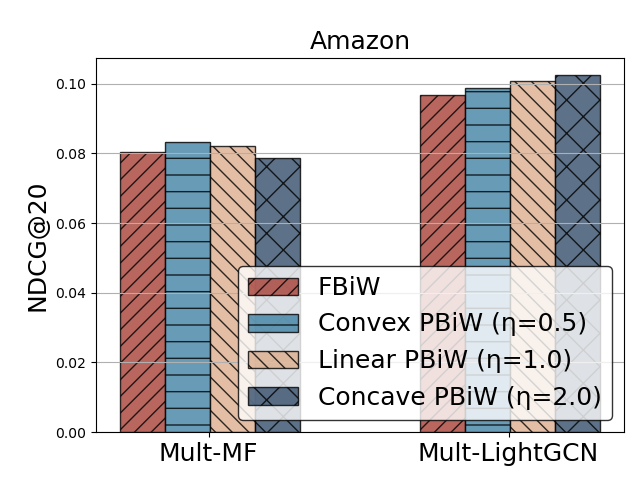}}
    \hfil
    \subfloat{\includegraphics[width=0.33\linewidth]{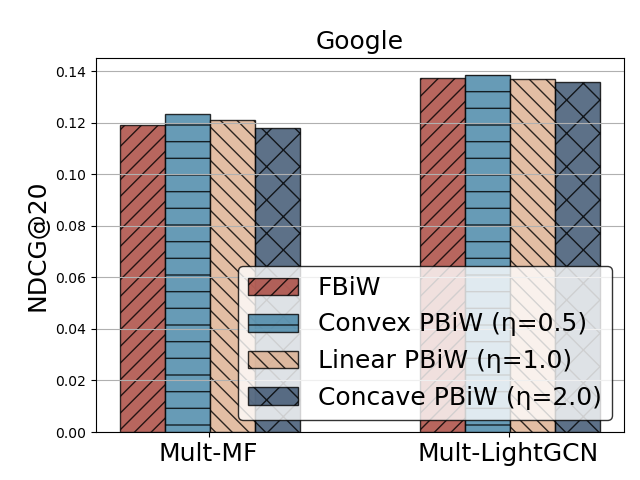}}
    \hfil
    \subfloat{\includegraphics[width=0.33\linewidth]{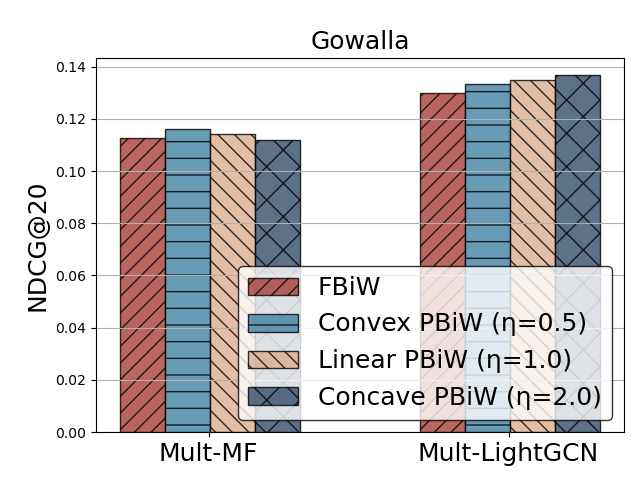}}
    \hfil
    \subfloat{\includegraphics[width=0.33\linewidth]{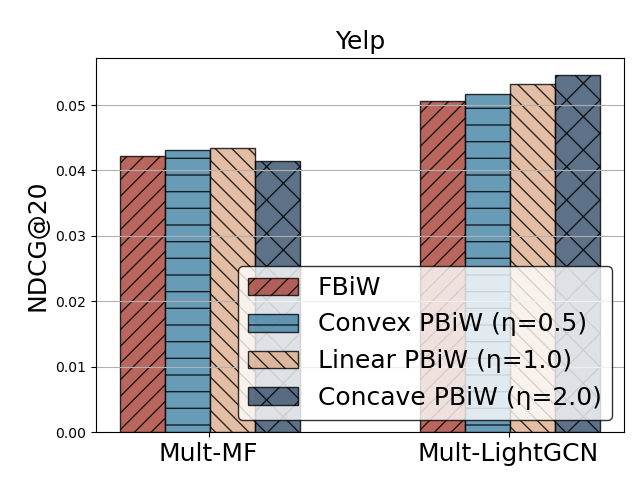}}
    \hfil
    \subfloat{\includegraphics[width=0.33\linewidth]{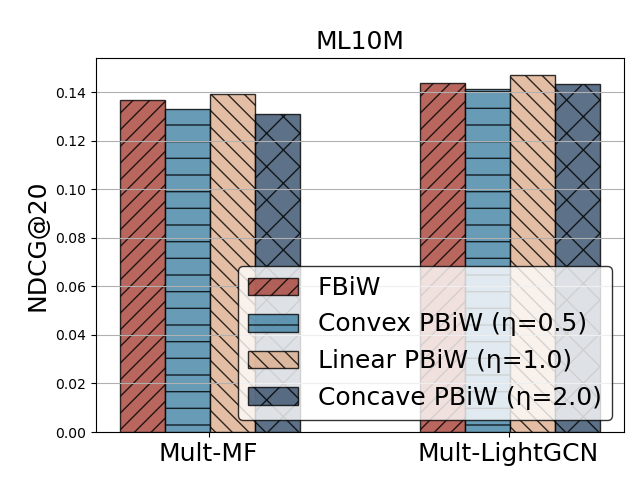}}
    \hfil
    \subfloat{\includegraphics[width=0.33\linewidth]{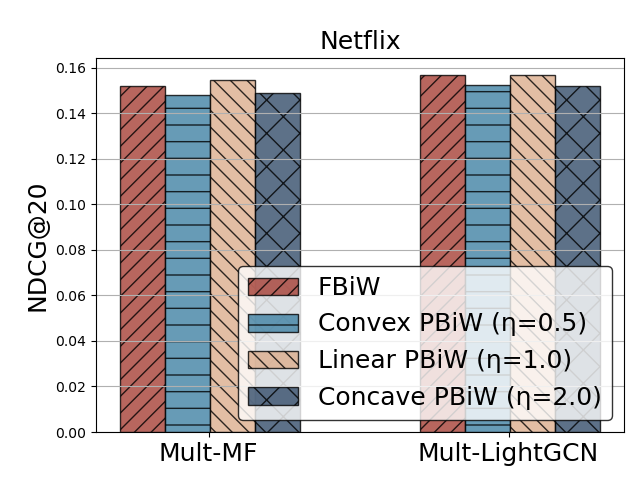}}
    \caption{Performance of Mult-MF and Mult-LightGCN with various Bi-Weighting strategies on six datasets.}
    \label{fig:eta}
\end{figure}

\section{Related Work}
\label{sec:rel}

\subsection{Collaborative Recommendation}
Research on recommender systems has evolved substantially since early collaborative filtering studies \cite{adomavicius2005toward}. Memory-based methods such as ItemKNN compute item similarity directly from interaction histories \cite{linden2003amazon,sarwar2001item}. Model-based approaches, most notably matrix factorization (MF), learn low-dimensional user and item representations and predict interactions via inner products \cite{koren2009matrix}. To alleviate data sparsity and better capture item-item relations, item-based models such as SLIM \cite{ning2011slim} and FISM \cite{kabbur2013fism} explicitly or implicitly learn item similarity structures.

Recent work has explored deep learning and graph-based architectures. NeuMF \cite{he2017neural} combines generalized matrix factorization with multilayer perceptrons to model nonlinear interactions, while NAIS \cite{he2018nais} extends item-based modeling by introducing attention mechanisms to differentiate the importance of historical items. Graph neural network approaches, such as NGCF \cite{wang2019neural}, model higher-order connectivity between users and items, and LightGCN \cite{he2020lightgcn} demonstrates that simplified graph convolutional designs can achieve strong performance by removing feature transformations and nonlinear activations. Self-supervised learning techniques, such as SGL \cite{wu2021self}, further enhance representation robustness through contrastive learning on user-item graphs. More recently, generative and diffusion-based models have been introduced for recommendation, with DiffRec \cite{wang2023diffusion} addressing efficiency and temporal preference shifts.

Despite these advances, most prior work focuses on improving recommendation accuracy through model design or training strategies, with limited consideration of popularity bias and its impact on recommendation fairness. Our work is complementary to these efforts, as we propose a model-agnostic debiasing method that can be applied to a wide range of existing recommenders.

\begin{figure}[!t]
    \centering
    \subfloat{\includegraphics[width=0.33\linewidth]{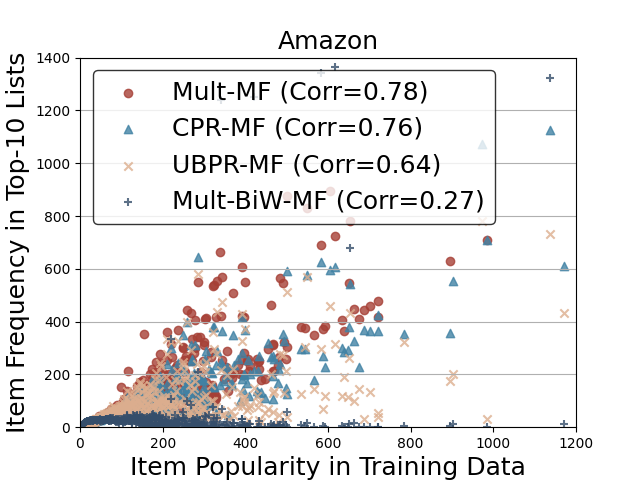}}
    \hfil
    \subfloat{\includegraphics[width=0.33\linewidth]{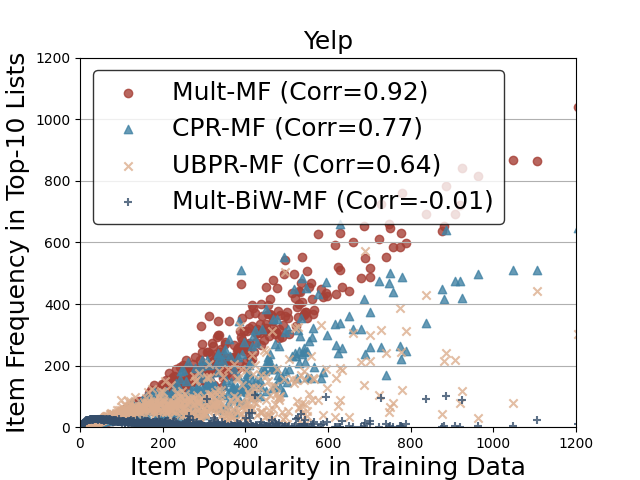}}
        \hfil
    \subfloat{\includegraphics[width=0.33\linewidth]{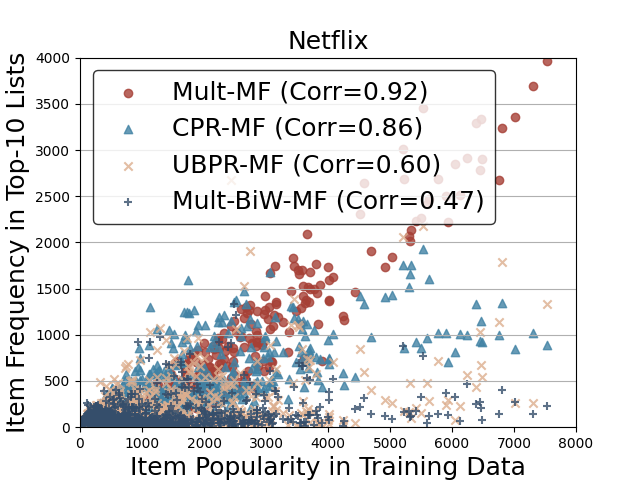}}
    \caption{Item frequency in top-10 recommendation lists vs. item popularity in training data on three datasets.}
    \label{fig:visualize}
\end{figure}

\begin{figure}[!t]
    \centering
    \subfloat{\includegraphics[width=0.33\linewidth]{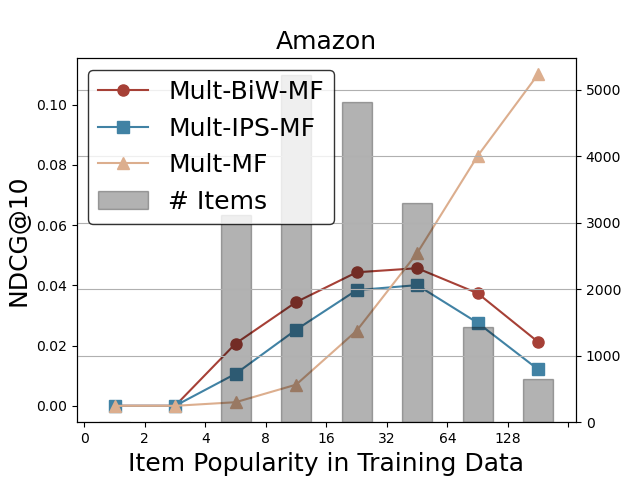}}
    \hfil
    \subfloat{\includegraphics[width=0.33\linewidth]{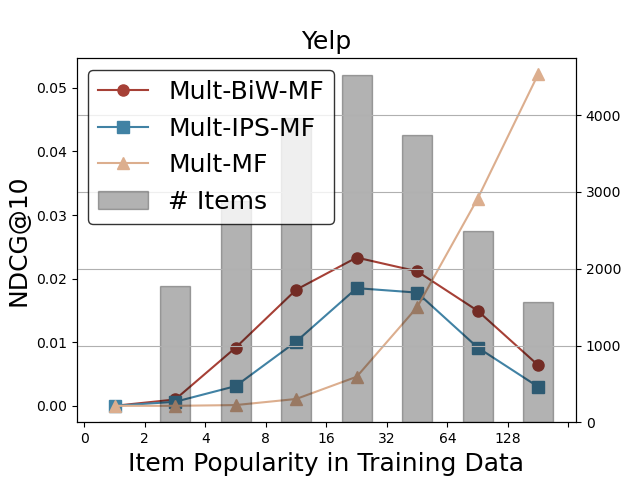}}
    \hfil
    \subfloat{\includegraphics[width=0.33\linewidth]{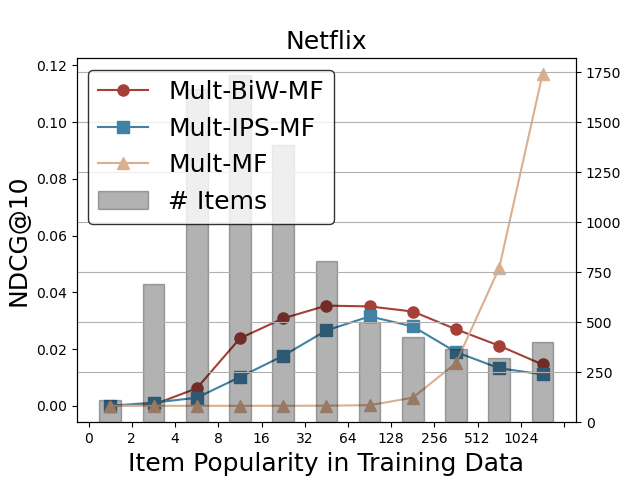}}
    \caption{Performance of Mult-MF, Mult-IPS-MF, and Mult-BiW-MF across item groups on three datasets.}
    \label{fig:group}
\end{figure}

\subsection{Loss Functions for Implicit Feedback}
Training recommendation models from implicit feedback data is challenging due to extreme sparsity and the absence of explicit negative signals. Existing approaches can be broadly categorized into pointwise, pairwise, and listwise loss functions.

Pointwise loss functions treat each user-item interaction independently and regress predicted scores to observed labels. Common examples include weighted squared loss \cite{hu2008collaborative,chen2020efficient2} and logistic or binary cross-entropy loss \cite{li2017neural}. To improve scalability, sampled binary cross-entropy variants draw one or more negative items for each positive instance \cite{johnson2014logistic,he2017neural,kang2018self,petrov2023gsasrec}. While efficient, pointwise objectives primarily capture local preferences at the individual interaction level.

Pairwise loss functions are designed to model relative preferences between items. BPR \cite{rendle2009bpr} constructs user-positive-negative triples and optimizes a ranking-based objective, while CML \cite{hsieh2017collaborative} employs a weighted hinge loss defined over embedding distances. Although pairwise approaches are effective for ranking with large item sets, they may suffer from slow convergence, sensitivity to negative sampling, and suboptimal global ranking performance \cite{he2016fast,rendle2014improving}.

Listwise loss functions directly optimize over an entire item list and are therefore better aligned with ranking metrics. For example, Mult-VAE \cite{liang2018variational} adopts a softmax cross-entropy loss to train variational autoencoders for recommendation, and sampled softmax losses have been proposed as efficient approximations for large-scale settings \cite{wu2024effectiveness}.

In the context of popularity-debiased recommendation, most existing methods rely on pointwise or pairwise losses \cite{saito2020unbiased,saito2020unbiased2,zheng2021disentangling,wan2022cross}, which may limit their ability to learn accurate global rankings under bias correction. In contrast, our work introduces a listwise debiasing framework that applies a multinomial likelihood to IPS-based methods, enabling more effective and efficient learning of balanced recommendations.

\subsection{Popularity Bias and Debiasing Strategies}
Popularity bias arises from the long-tailed distribution of user feedback, where a small number of popular items account for a large fraction of interactions \cite{abdollahpouri2020multi}. Models trained on such data tend to inherit and amplify this imbalance, resulting in the over-recommendation of popular items \cite{chen2023bias}. Existing mitigation methods are commonly grouped into pre-processing, in-processing, and post-processing approaches, depending on the stage at which bias is addressed.

Pre-processing methods aim to debias recommendation by modifying the training data or reweighting samples before model learning. A prominent line of work is based on inverse propensity scoring (IPS), which down-weights over-exposed interactions to obtain unbiased estimators \cite{schnabel2016recommendations,ma2019missing}. Representative methods include Rel-MF \cite{saito2020unbiased} and UBPR \cite{saito2020unbiased2}. AutoDebias \cite{chen2021autodebias} further combines IPS with data imputation and meta-learning to automatically identify effective debiasing strategies. Closely related approaches include doubly robust and imputation-based methods that integrate propensity scores with error correction. DRJL \cite{wang2019doubly} alternates between training a prediction model and an imputation model, while CDR \cite{song2023cdr} improves robustness by filtering unreliable imputed values. More recent work extends these ideas to graph-based recommenders, such as DR-GNN \cite{wang2024distributionally}, and addresses exposure bias by distinguishing between unclicked items due to user disinterest or lack of exposure, as in ReCRec \cite{lin2024recrec}.

In-processing methods mitigate popularity bias by modifying model architectures, regularization terms, or learning objectives. Zerosum \cite{rhee2022countering} introduces a regularization term that penalizes score differences among items that are equally preferred by a user. ReSN \cite{lin2025recommendation} observes that popularity bias is encoded in the dominant spectrum of the score matrix and alleviates it through spectral norm regularization. PAAC \cite{cai2024popularity} employs popularity-aware alignment and reweighted contrastive learning to reduce representation separation. Another line of work leverages causal reasoning to disentangle intrinsic user interest from popularity conformity. CausE \cite{bonner2018causal} learns separate embeddings from biased and unbiased data and aligns them via regularization. DICE \cite{zheng2021disentangling} explicitly separates interest and conformity factors, while MACR \cite{wei2021model} adopts a counterfactual multi-task framework to remove the direct effect of item popularity. CPR \cite{wan2022cross} derives an unbiased pairwise ranking loss based on carefully selected training pairs.

Post-processing methods operate at inference time and adjust recommendation scores or rankings without retraining models. These include score penalization or normalization strategies that suppress popular items \cite{abdollahpouri2017controlling}, reranking approaches that promote long-tail items \cite{zhu2021popularity}, and embedding normalization techniques such as TTEN \cite{kim2023test}. Beyond bias mitigation, several studies examine the role of popularity signals themselves. Some efforts argue that popularity can be exploited to improve recommendation utility rather than strictly eliminated \cite{zhang2021causal,zhao2022popularity}, while others analyze bias dynamics through simulation and temporal modeling \cite{zhu2021popularity2}.

Among existing debiasing techniques, IPS-based methods are particularly attractive due to their strong theoretical guarantees. However, their effectiveness critically depends on accurate propensity estimation. To address this limitation, we propose a Bi-Weighting strategy to improve propensity estimation, and we integrate it with a multinomial, listwise IPS objective to generate more accurate and balanced recommendations.

\section{Conclusion}
\label{sec:con}

In this paper, we examined the challenges of applying Inverse Propensity Scoring (IPS) to popularity debiasing in item recommendation and proposed a novel method, termed Multinomial Likelihood with Bi-Weighting (Mult-BiW). We first introduced a listwise unbiased learning framework, Mult-IPS, which integrates IPS with a multinomial likelihood to capture global and unbiased user preferences over the entire item set. Building on this framework, we proposed a Bi-Weighting (BiW) strategy that combines individual propensity estimates with a collection-level model, introducing a principled smoothing mechanism to improve the robustness of propensity estimation. We further provided theoretical analyses that establish an upper bound on the empirical bias induced by inaccurate propensities and derived the optimal form of the collection model underlying BiW. To mitigate the adverse effects of aggressive reweighting on representation learning, we designed a progressive Bi-Weighting strategy based on cumulative learning, enabling a smooth transition from representation learning to popularity debiasing. Extensive experiments on six real-world datasets demonstrate that Mult-BiW consistently outperforms state-of-the-art debiasing methods by a substantial margin, while effectively mitigating popularity bias and producing more balanced recommendation outcomes. These results validate the effectiveness of combining listwise unbiased learning with robust propensity smoothing.

Several promising directions remain for future work. First, more advanced smoothing techniques, such as Bayesian smoothing and absolute discounting, could be explored to further enhance propensity estimation robustness \cite{zhai2004study}. Second, incorporating auxiliary side information, such as user exposure signals, viewing histories, or click-through data, may enable more expressive Bi-Weighting strategies and improve debiasing performance in practical systems \cite{gao2022kuairand}. We leave these extensions for future investigation.

\begin{acks}
This work was supported by the National Natural Science Foundation of China (Grant No. 72401015 to Tianyu Zhu and Grant No. 72501067 to Yansong Shi) and the China Postdoctoral Science Foundation (Grant Nos. GZC20242152 and 2024M764106 to Tianyu Zhu).
\end{acks}



\bibliographystyle{ACM-Reference-Format}
\bibliography{main}

\end{document}